\RequirePackage[figuresleft]{rotating}
\documentclass[11pt]{amsart}
\usepackage{amsmath, amsfonts, amsthm, amssymb}
\usepackage{fullpage}
\usepackage{paralist}
\usepackage{graphics} 
\usepackage{epsfig} 
\usepackage{graphicx}  
\usepackage{multicol}
\usepackage{epstopdf}
\usepackage{xcolor}
\usepackage{colortbl}
\usepackage{blkarray}
\usepackage{url}
\usepackage[table]{xcolor}
\usepackage{pdfpages}
\usepackage[foot]{amsaddr}
\usepackage[colorlinks=true]{hyperref}
\hypersetup{urlcolor=blue, citecolor=red}

\usepackage{tikz} 
\usetikzlibrary{shapes,arrows, patterns,shapes.geometric}

\newtheorem{theorem}{Theorem}[section]
\newtheorem{corollary}[theorem]{Corollary}

\newtheorem{proposition}[theorem]{Proposition}
\newtheorem{construction}{Construction}
\newtheorem*{construction*}{Construction}

\theoremstyle{definition}
\newtheorem{assumption}[theorem]{Assumption}

\newtheorem{definition}[theorem]{Definition}

\newtheorem{example}[theorem]{Example}

\newtheorem{remark}[theorem]{Remark}

\newcommand{\C}{\mathcal{C}}
\newcommand{\F}{\mathbb{F}}

\tikzstyle{b} = [draw, thick, black, -]
\tikzstyle{d} = [draw, thick,black, dashed]

\title{Hierarchical Quasi-cyclic Codes from Reed-Solomon and Polynomial Evaluation Codes}

\author{Emily McMillon$^1$}
\address{$^1$Department of Mathematics, Rice University, Houston, TX 77005}
\email{em72@rice.edu}
\thanks{The first author was supported by the National Science Foundation grant DMS-2303380.}

\author{Kathryn Haymaker$^2$}
\address{$^2$Department of Mathematics, Villanova University, 
Villanova, PA 19085}
\email{kathryn.haymaker@villanova.edu}



\begin{document}

\begin{abstract} 
    We introduce the first example of algebraically constructed hierarchical quasi-cyclic codes. These codes are built from Reed-Solomon codes using a 1964 construction of superimposed codes by Kautz and Singleton. We show both the number of levels in the hierarchy and the index of these Reed-Solomon derived codes are determined by the field size. We show that this property also holds for certain additional classes of polynomial evaluation codes.
    
    We provide explicit code parameters and properties as well as some additional bounds on parameters such as rank and distance. In particular, starting with Reed-Solomon codes of dimension $k=2$ yields hierarchical quasi-cyclic codes with Tanner graphs of girth $6$.

    We present a table of small code parameters and note that some of these codes meet the best known minimum distance for binary codes, with the additional hierarchical quasi-cyclic structure. We draw connections to similar constructions in the literature, but importantly, while existing literature on related codes is largely simulation-based, we present a novel algebraic approach to determining new bounds on parameters of these codes.
\end{abstract}

\maketitle

\section{Introduction} \label{sec:intro}

Parity-check codes are linear codes specified by a parity-check matrix, which plays a central role in determining code parameters and often forms the basis for a decoding strategy. Notable families include low-density parity-check (LDPC) codes, moderate-density parity-check (MDPC) codes, and codes defined via sparse Tanner graph representations. LDPC codes and their iterative decoding algorithm were introduced by Gallager \cite{gallager1960}. Modern research in MDPC codes is inspired by their role in post-quantum cryptography schemes based on the McEliece cryptosystem \cite{misoczki2013mdpc}.

Here, we explore an algebraic construction of parity-check codes motivated by a study of moderate-density and low-density parity-check codes from disjunct matrices \cite{haymaker2025parity}. This construction is inspired by Kautz and Singleton's classical work on superimposed codes \cite{kautz1964nonrandom}. The Kautz and Singleton approach  has seen renewed interest recently (e.g. \cite{inan2019optimality, essential20}), including a flurry of new work related to group testing for COVID-19 \cite{shental2020efficient, lau2021construction,haymaker2025pooled, dinesen2025group}.

Kautz and Singleton used $q$-ary codes to create binary superimposed codes, and since there are many ways to derive binary codes from nonbinary, we will delineate the differences between our main construction and other methods. One notable related construction is for LDPC codes from 2-dimensional MDS codes \cite{chen2010two}, a special case of which was adopted as the standard for 10 Gigabit Ethernet over copper (10GBASE-T) \cite{lim2012novel}. Related LDPC codes from 2-dimensional MDS codes are presented in \cite{djurdjevic2003class, xu2006construction}. A variation called mirror-paradigm Reed-Solomon based quasi-cyclic LDPC codes was proposed for nonvolatile memory channels in \cite{lim2012novel}. More recently, a correspondence that sends rank-metric codes to Hamming codes via a similar length $q$ unit-vector mapping was studied in \cite{astore2024geometric}. 

There are different methods of mapping non-binary Reed-Solomon (or other) codes into binary. The \textit{binary image} of a code over $\mathbb{F}_{2^m}$ replaces each $\mathbb{F}_{2^m}$ symbol with a binary $m$-tuple (see \cite{macwilliams1977theory}). For example, the binary expanded Reed-Solomon codes studied in \cite{bidan2007some} replace each field element from $\mathbb{F}_{2^m}$ with a binary $m$-tuple. In contrast, the mapping in this paper sends symbols from an arbitrary field $\mathbb{F}_q$ to binary $q$-tuples of weight one. The resulting codes are quasi-cyclic, with index  equal to the characteristic of the underlying field, and have further hierarchical quasi-cyclic properties that we prove.

Quasi-cyclic parity-check codes are among the most widely studied classes of structured linear codes due to their simple and practical implementations (e.g. \cite{karimi2013girth, wang2013hierarchical, myung2005quasi, li2018algebra}). A code is quasi-cyclic if it is invariant under shifts by some integer $t$, called the index. Quasi-cyclic LDPC and MDPC codes have proven to be successful and produce codes with predictable rank behavior and favorable Tanner graph properties while maintaining efficient implementations \cite{fossorier2004quasicyclic}. These advantages have led to the use of quasi-cyclic LDPC codes in a variety of modern communication and storage standards. The rank of QC-LDPC codes has been studied in \cite{yang2018rank, smarandache2022using} and references therein.

Hierarchical quasi-cyclic (HQC) codes were introduced in \cite{wang2013hierarchical} as a generalization of quasi-cyclic codes. In these codes, the matrix structure is globally quasi-cyclic, and the smaller pieces of the code are also quasi-cyclic, potentially with multiple such levels. This also generalizes the idea of taking the pre-lifted protograph approach, as in e.g. \cite{mitchell2014quasi, li2025high}, to an arbitrary number of quasi-cyclic graph lifts. From a practical perspective, HQC codes present a way to compactly represent more complicated, but still structured, matrices. From a theoretical perspective, the presence of multiple quasi-cyclic layers introduces parameters\textemdash such as the number of levels and associated indices\textemdash that can be leveraged to influence code parameters.

In \cite{haymaker2025parity}, it was shown that by taking a Reed-Solomon code of length $n$ and dimension $2$, the parity-check codes obtained via the construction considered in this paper are free of $4$-cycles, which is advantageous for iterative decoding. However, specific parameters such as exact dimension and minimum distance were left as open questions. 

In this work, we show that hierarchical quasi-cyclic structure can arise naturally from an algebraic construction based on Reed–Solomon codes and other classes of polynomial evaluation codes via the Kautz–Singleton superimposed code construction. To our knowledge, this is the first explicit algebraic construction yielding hierarchical quasi-cyclic codes. This means that the code parameters are not imposed externally, but are intrinsic to the algebraic construction, and enables analysis of their parameters using algebraic tools. A few of the results were first presented in \cite{HM25}, some without proofs. This version has significantly more background and motivation and also contains new results on the hierarchical quasi-cyclic structure of parity-check codes from Reed-Solomon and certain polynomial evaluation codes, using the main construction. 

The paper is structured as follows: Section~\ref{sec:prelim} introduces standard coding theory definitions as well as needed background on hierarchical quasi-cyclic codes and the key constructions. The parameters and an algebraic analysis of our main construction are given in Section~\ref{sec:parameters} as well as a characterization of these codes within the hierarchical quasi-cyclic framework. Results on the relationship between the main construction and existing codes, as well as the impact of permutation and monomial equivalence in the starting $q$-ary codes on the resulting binary codes, are proven in Section~\ref{sec:equivalence}. Section~\ref{sec:table} gives a table of code parameters for small field sizes. Section~\ref{sec:conclusion} concludes the paper.

\section{Preliminaries} 
\label{sec:prelim}

We begin with some general notation. Throughout, $[N]$ denotes the set $\{1, 2, \dotsc, N\}$. The finite field with $q$ elements is denoted $\mathbb{F}_q$ ($q$ is always a power of a prime). Vectors are denoted in boldface, i.e. $\mathbf{x}$. The support of a binary string $\mathbf{x}$ is the set of indices $i$ where $x_i=1$.  Finally, define $[\mathbf{e}_i]_n$ to be the $i$th standard basis column vector of length $n$, and define $[\mathbf{u}_i]_n$ to be the $i$th standard basis row vector of length $n$. If the length of the vectors is clear from context, we use shorthand $\mathbf{e}_i$ and $\mathbf{u}_i$, respectively. 

\subsection{Coding theory fundamentals}

An $[n,k]_q$ \textit{linear error-correcting code} is a $k$-dimensional subspace of the $n$-dimensional vector space $\mathbb{F}_q^n$. This subspace can be defined as the kernel of an $m \times n$ matrix $H$ with row rank $n-k$, called a \textit{parity-check matrix} of the code. That is, given a parity-check matrix $H$, the code $C(H)$ is given by
\[ C(H) = \{ \mathbf{x} \in \mathbb{F}_q^n \mid H \mathbf{x}^T = \mathbf{0}\}. \]

Codes can also be defined using a \textit{generator matrix} $G$ such that $G H^T = 0$. In this case, the code $C(G)$ is given by the row space of $G$, or 
\[ C(G) = \{ \mathbf{c} \in \F_q^n \mid \mathbf{c} = \mathbf{x} G, \mathbf{x} \in \F_q^k \}.\]

The \textit{minimum distance} of a linear code $C$ is $d_{\text{min}}(C) = \min \limits_{\mathbf{c} \in C} \text{wt}_H(\mathbf{c})$ where $\text{wt}_H(\mathbf{c})$ is the \textit{Hamming weight}, or number of nonzero coordinates, of $\mathbf{c}$. If the minimum distance of a linear code $C$ is known, it is called an $[n,k,d]_q$ code. A well-known and often useful fact is that the minimum distance of a code $C(H)$ is the smallest size of a collection of linearly dependent columns of $H$.

\begin{definition} Let $\mathbb{F}_q$ be a finite field with elements labeled $\alpha_1, \ldots, \alpha_q$. A \textit{Reed-Solomon (RS) code} $C$ over $\mathbb{F}_q$ with parameters $[n, k, n-k+1]_q$, $n\leq q$,  is 
\[ C:= \{\mathbf{c}_p \mid p(x)\in \mathbb{F}_q[x], \deg(p(x))<k\}, \] 
where the codeword $\textbf{c}_p$ from the polynomial $p(x)$ is  
\[ \textbf{c}_p:=(p(\alpha_1), \ldots, p(\alpha_n)).\]
When $n=q$, the code is called an \textit{extended narrow-sense Reed-Solomon code}.
\end{definition}

We remark that Reed-Solomon codes are \textit{maximum distance separable} (MDS) since $d=n-k+1$, the maximum possible distance for fixed $n$ and $k$.

\begin{definition}
    The \textit{Reed-Muller code} $RM(q, m, \rho)$ is the code of evaluations over $\mathbb{F}_q^m$ of $q$-ary polynomials in $m$ variables of total degree at most $\rho$ and individual degree at most $q-1$. Let $\mathcal{P}=\mathbb{F}_q[X_1, X_2, \ldots, X_m]$, and $\mathcal{P}_{\leq \rho}$ be the set of polynomials of total degree at most $\rho$.  That is, 
    \[RM(q, m, \rho)= \{(f(\alpha_1), \ldots, f(\alpha_{q^m}))|f\in \mathcal{P}_{\leq \rho},  \mathbb{F}_q^m=\{\alpha_1, \ldots, \alpha_{q^m}\}\}. \] 
\end{definition}

Throughout, we will order the elements of $\mathbb{F}_{p^r}$ and the polynomials in both $\mathbb{F}_{p^r}[x]$ and $\mathbb{F}_{p^r}[X_1, \ldots, X_m]$ using $p$-ary lexicographic ordering. 

\begin{definition} \label{def:pary}
    Let $\F_q$ with $q = p^r$ be a finite field of characteristic $p$ and let $\alpha$ be a primitive element of $\F_q$. Then any element of $\F_q$ can be written as $a_0 + a_1 \alpha + a_2 \alpha^2 + \dots + a_{r-1}\alpha^{r-1}$ for some collection of $a_i \in \F_p$. We will induce a \textit{$p$-ary lexicographic ordering} on these elements where $0 < 1 < \dots < p-1$ and $\alpha^0 < \alpha^1 < \dots < \alpha^{r-1}$. Similarly, for elements in $\F_q[x]$, order the coefficients of the polynomials in the $p$-ary lexicographic ordering and $x^0 < x^1< x^2 < \cdots$. For elements in $\mathbb{F}_{q}[X_1, \ldots, X_m]$,  order the coefficients of the polynomials in the $p$-ary lexicographic ordering, then use graded lexicographic ordering on the polynomials, first comparing total degree, then using lexicographic ordering on the monomial terms. 
\end{definition}

We next define a property of both Reed-Solomon and Reed-Muller codes, which is also enjoyed by certain subcodes of these polynomial evaluation codes. 

\begin{definition} \label{def:fieldpartition}
    Suppose $\mathbb{F}_{p^r}$ is a field with elements $0, 1,2, \ldots,p-1, \alpha, \ldots, (p-1)(\alpha^{r-1}+\alpha^{r-2}+\cdots +1) $ in $p$-ary lexicographic order. A polynomial evaluation code $C$ over $\mathbb{F}_{p^r}$ has the \textit{field partition property} if the set of polynomials defining $C$ can be partitioned into $\frac{|C|}{p^r}$ parts, each with representative $q_i(x)$ with zero constant term, where each part consists of the $p^r$ constant shifts of $q_i$: $\{q_i, q_i+1, q_i+2,\ldots, q_i+(p-1)(\alpha^{r-1}+\alpha^{r-2}+\cdots +1)\}$.
\end{definition}

\subsection{Quasi-cyclic and hierarchical quasi-cyclic codes}

We now define quasi-cyclic codes. Let $\mathbf{x} = (x_1, \dotsc, x_n)$ be an arbitrary $n$-tuple. Let $\mathbf{x}^i$ be the cyclic shift of $\mathbf{x}$ $i$ positions to the right, i.e., $\mathbf{x}^i = (x_{n-i+1}, x_{n-i}, \dotsc, x_{n}, x_1, \dotsc, x_{n-i})$.

\begin{definition}
    A length $n$ \textit{quasi-cyclic code} is a linear code for which, under some representation of the code, cyclically shifting a codeword a fixed number $L > 1$ (or a multiple of $L$) of symbol positions either to the right or to the left results in another codeword. The integer $L$ is called the \textit{shifting constraint}. The integer $\frac{n}{L}=t$ is called the \textit{index}. When $L=1$, the code is called \textit{cyclic}.
\end{definition}

Given shifting constraint $L$ and positive integers $J, t$ with $J < t$, let $H_i \in \F_q^{J \times L}$ for $i \in [t]$. The nullspace of the matrix $H \in \F_q^{tJ \times t L}$ in Equation~\ref{eq:blockcirculant}, which is said to be in \textit{block circulant form}, is a quasi-cyclic parity-check code. 

\begin{equation}
   H = \begin{bmatrix} H_1 & H_2 & \dotsm & H_t \\
                       H_t & H_1 & \dotsm & H_{t-1} \\
                       \vdots & \vdots & \ddots & \vdots \\
                       H_2 & H_3 & \dotsm & H_1 \end{bmatrix}
                       \label{eq:blockcirculant}
\end{equation}

\begin{equation}
   \hat{H} = \begin{bmatrix} C_{1,1} & C_{1,2} & \dotsm & C_{1,L} \\
                         C_{2,1} & C_{2,2} & \dotsm & C_{2,L} \\
                         \vdots & \vdots & \ddots & \vdots \\
                         C_{J,1} & C_{J,2} & \dotsm & C_{J,L} \end{bmatrix}
                         \label{eq:circulant}
\end{equation}

The code from such a matrix $H$ is permutation equivalent to a code derived from a matrix $\hat{H}$ consisting of circulant matrices. A \textit{circulant matrix} is square matrix such that each successive row is the row above it shifted cyclically one position to the right. In this case, $\hat{H}$ is a $J \times L$ block matrix composed of $t \times t$ circulants, as in Equation~\ref{eq:circulant}.

Throughout, we will denote the $t \times t$ identity matrix which has been cyclically shifted $r$ positions to the left as $I_{t,r}$. Any circulant matrix can be represented as some linear combination of such shifted identity matrices.

In general, the parity-check matrix of a quasi-cyclic code, as given in Equation~\ref{eq:circulant}, can be represented by a \textit{polynomial parity-check matrix} $H(x)$, with entries that are polynomials:
\begin{equation}
    H(x) = \begin{bmatrix} h_{1,1}(x) & h_{1,2}(x) & \dots & h_{1,L}(x) \\
    h_{2,1}(x) & h_{2,2}(x) & \dots & h_{2,L}(x) \\
    \vdots & \vdots & \ddots & \vdots \\
    h_{J,1}(x) & h_{J,2}(x) & \dots & h_{J,L}(x)\end{bmatrix}. \label{eq:polyh}
\end{equation}

There is a natural connection between the representation in Equation~\ref{eq:polyh} and the circulant representation in Equation~\ref{eq:circulant}. We let the first row of each circulant submatrix $C_{j,\ell}$ correspond to the vector representation of the coefficients of $h_{j,\ell}(x)$. A common other interpretation (when the codes are binary and assuming $t$ is well understood) is to think of $x^r$ as representing $I_{t,r}$. The polynomial interpretation is thus a compact way of representing sums of circulant matrices.

Quasi-cyclic LDPC codes are often characterized by the weight of the highest weight circulant, or equivalently, the largest entry in the code's weight matrix. In general, if the highest weight circulant has weight $T$, a code is called a Type-$T$ quasi-cyclic code. The \textit{weight matrix} $W$ of a quasi-cyclic parity-check matrix is the number of nonzero entries in the first row of each circulant block, or, equivalently, the number of nonzero terms in each polynomial entry of $H(x)$. An example of these terms follows.

\begin{example} \label{ex:polynomial-circulants}
    Let $C$ be a binary quasi-cyclic parity-check code with the following parity-check matrix.
    \begin{equation} \label{eq:qc-matrix}
    {H = \left[
        \begin{array}{cccc|cccc|cccc}
            1 & 0 & 1 & 0 & 0 & 0 & 0 & 1 & 0 & 0 & 0 & 0 \\
            0 & 1 & 0 & 1 & 1 & 0 & 0 & 0 & 0 & 0 & 0 & 0 \\
            1 & 0 & 1 & 0 & 0 & 1 & 0 & 0 & 0 & 0 & 0 & 0 \\
            0 & 1 & 0 & 1 & 0 & 0 & 1 & 0 & 0 & 0 & 0 & 0 \\
            \hline
            1 & 0 & 0 & 0 & 0 & 1 & 0 & 0 & 0 & 0 & 1 & 0 \\
            0 & 1 & 0 & 0 & 0 & 0 & 1 & 0 & 0 & 0 & 0 & 1 \\
            0 & 0 & 1 & 0 & 0 & 0 & 0 & 1 & 1 & 0 & 0 & 0 \\
            0 & 0 & 0 & 1 & 1 & 0 & 0 & 0 & 0 & 1 & 0 & 0 \\
        \end{array} \right]_{8 \times 12}}
    \end{equation}
    Here, the index is $t=4$, we have $J=2$, $L=3$, and the polynomial version of the parity-check matrix is given by
    \begin{equation} \label{eq:poly-matrix}
    H(x) = \begin{bmatrix} 1+x^2 & x & 0 \\ 1 & x^3 & x^2
    \end{bmatrix}.
    \end{equation}
    The weight matrix of this parity-check matrix is
    \begin{equation} \label{eq:wtmatrix}
        W = \begin{bmatrix} 2 & 1 & 0 \\ 1 & 1 & 1 \end{bmatrix}.
    \end{equation}
    and so this code is type-2.
\end{example}

The weight matrix can be used as the biadjacency matrix of a graph, called a \textit{Tanner graph}. In this way, the rows correspond to check nodes (denoted with squares) and the columns correspond to variable nodes (denoted with circles). The $w_{j,\ell}$ entry corresponds to the number of edges between check node $c_j$ and variable node $v_{\ell}$. Such a graph is usually called a \textit{protograph}. The protograph corresponding to the matrix in Equation~\ref{eq:wtmatrix} is given in Figure~\ref{fig:basegraph}.

\begin{figure}[ht]
    \centering
    \begin{tikzpicture}[square/.style={regular polygon,regular polygon sides=4}]
        \node[circle,draw=black,fill=white] (v1) at (0,0) {};
        \node[circle,draw=black,fill=white] (v2) at (2,0) {};
        \node[circle,draw=black,fill=white] (v3) at (4,0) {};
        \node[square,draw=black,fill=white] (c1) at (1,2) {};
        \node[square,draw=black,fill=white] (c2) at (3,2) {};

        \draw[black,thick] (v1) to [out=75,in=-125] (c1);
        \draw[black,thick] (v1) to [out=55,in=-100] (c1);
        \draw[black,thick] (v1) to (c2);
        \draw[black,thick] (v2) to (c1);
        \draw[black,thick] (v2) to (c2);
        \draw[black,thick] (v3) to (c2);
    \end{tikzpicture}
    \caption{The base Tanner graph corresponding to the weight matrix in Equation~\ref{eq:wtmatrix}.} \label{fig:basegraph}
\end{figure}
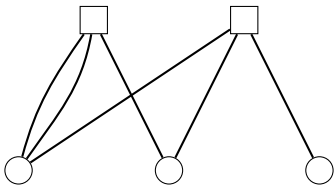

The base graph is used to create quasi-cyclic codes via a \textit{graph lifting} operation. Each vertex in the base graph is replaced with $t$ vertices in the lifted graph. In this case $t$ is called the \textit{lifting degree} of the graph lift. Edges are maintained between the copies, but may be permuted. When these permutations are cyclic, a quasi-cyclic parity-check matrix is obtained. The Tanner graph in Figure~\ref{fig:liftedgraph} has been obtained by starting with the weight matrix in Equation~\ref{eq:wtmatrix} and lifting the edges according the polynomial parity-check matrix in Equation~\ref{eq:poly-matrix} or, equivalently, the parity-check matrix in Equation~\ref{eq:qc-matrix}.

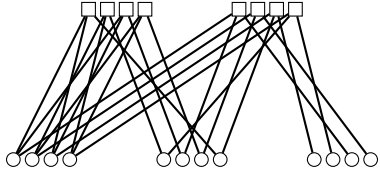
\begin{figure}[ht]
    \centering
    \begin{tikzpicture}[square/.style={regular polygon,regular polygon sides=4}]
        \node[circle,draw=black,fill=white,scale=0.5] (v11) at (0,0) {};
        \node[circle,draw=black,fill=white,scale=0.5] (v12) at (0.25,0) {};
        \node[circle,draw=black,fill=white,scale=0.5] (v13) at (0.5,0) {};
        \node[circle,draw=black,fill=white,scale=0.5] (v14) at (0.75,0) {};
        \node[circle,draw=black,fill=white,scale=0.5] (v21) at (2,0) {};
        \node[circle,draw=black,fill=white,scale=0.5] (v22) at (2.25,0) {};
        \node[circle,draw=black,fill=white,scale=0.5] (v23) at (2.5,0) {};
        \node[circle,draw=black,fill=white,scale=0.5] (v24) at (2.75,0) {};
        \node[circle,draw=black,fill=white,scale=0.5] (v31) at (4,0) {};
        \node[circle,draw=black,fill=white,scale=0.5] (v32) at (4.25,0) {};
        \node[circle,draw=black,fill=white,scale=0.5] (v33) at (4.5,0) {};
        \node[circle,draw=black,fill=white,scale=0.5] (v34) at (4.75,0) {};
        \node[square,draw=black,fill=white,scale=0.5] (c11) at (1,2) {};
        \node[square,draw=black,fill=white,scale=0.5] (c12) at (1.25,2) {};
        \node[square,draw=black,fill=white,scale=0.5] (c13) at (1.5,2) {};
        \node[square,draw=black,fill=white,scale=0.5] (c14) at (1.75,2) {};
        \node[square,draw=black,fill=white,scale=0.5] (c21) at (3,2) {};
        \node[square,draw=black,fill=white,scale=0.5] (c22) at (3.25,2) {};
        \node[square,draw=black,fill=white,scale=0.5] (c23) at (3.5,2) {};
        \node[square,draw=black,fill=white,scale=0.5] (c24) at (3.75,2) {};

        \draw[black,thick] (v11) to (c11);
        \draw[black,thick] (v12) to (c12);
        \draw[black,thick] (v13) to (c13);
        \draw[black,thick] (v14) to (c14);
        \draw[black,thick] (v11) to (c13);
        \draw[black,thick] (v12) to (c14);
        \draw[black,thick] (v13) to (c11);
        \draw[black,thick] (v14) to (c12);
        \draw[black,thick] (v11) to (c21);
        \draw[black,thick] (v12) to (c22);
        \draw[black,thick] (v13) to (c23);
        \draw[black,thick] (v14) to (c24);
        \draw[black,thick] (v21) to (c12);
        \draw[black,thick] (v22) to (c13);
        \draw[black,thick] (v23) to (c14);
        \draw[black,thick] (v24) to (c11);
        \draw[black,thick] (v21) to (c24);
        \draw[black,thick] (v22) to (c21);
        \draw[black,thick] (v23) to (c22);
        \draw[black,thick] (v24) to (c23);
        \draw[black,thick] (v31) to (c23);
        \draw[black,thick] (v32) to (c24);
        \draw[black,thick] (v33) to (c21);
        \draw[black,thick] (v34) to (c22);
    \end{tikzpicture}
    \caption{The quasi-cyclic lift of the graph in Figure~\ref{fig:basegraph} according to the parity-check matrices given in Example~\ref{ex:polynomial-circulants}.} \label{fig:liftedgraph}
\end{figure}

In \cite{wang2013hierarchical}, the notion of a quasi-cyclic code using an arbitrary number of nested graph liftings was introduced. Codes that can be represented this way are called \textit{hierarchical quasi-cyclic} or HQC.

\begin{definition}[\cite{wang2013hierarchical}] \label{def:hqc}
    A \textit{hierarchical quasi-cyclic LDPC code} with $K$ levels is defined by a $J_K \times L_K$ multivariate polynomial parity-check matrix $H(x_1, \dotsc, x_K)$ in $K$ variables. The $(j,\ell)$ entry of $H$, with $j \in [J_K]$, $\ell \in [L_K]$, is given by a $K$-variate polynomial $h_{j,\ell}(x_1, \dotsc, x_K)$. The value $t_i$ for $i \in [K]$ is the $i$th lifting degree, so for all terms in the matrix, the maximum degree on $x_i$ is $t_i-1$. Define $c_{s_1,\dotsc, s_K}^{j,\ell}$ to be the coefficient associated with the term $x_1^{s_1} x_2^{s_2} \dotsm x_K^{s_K}$. With these notation and definitions, the code is defined by the collection of $J_K L_K$ polynomials:
    \begin{equation} \label{eq:hqc-polynomials}
        h_{j,\ell}(x_1, \dotsc, x_K) = \sum_{s_K=0}^{t_K-1} \dots \sum_{s_1=0}^{t_1-1} c_{s_1, \dotsc, s_K}^{j,\ell} \left( \prod_{k=1}^K x_k^{s_k} \right).
    \end{equation}
    The parity-check matrix of this code is obtained by replacing each entry of $H(x_1, \dotsc, x_K)$ with the submatrix 
    \begin{equation} \label{eq:hqc-pcmatrix}
        \sum_{s_K=0}^{t_K-1} \dotsm \sum_{s_1=0}^{t_1-1} c_{s_1, \dotsc, s_K}^{j,\ell} \left( I_{t_K,s_K} \otimes \dotsm \otimes I_{t_1,s_1} \right)
    \end{equation}
    where $\otimes$ is the Kronecker product.
\end{definition}

Note that the Kronecker product is not commutative. Normally, the ordering of variables in a multivariate polynomial $h(x_1, \dotsc, x_K)$ over $\F_q$ is not consequential. However, in this setting, changing the order of the variables in Equation~\ref{eq:hqc-polynomials} changes the resulting matrix representation given by the Kronecker product in Equation~\ref{eq:hqc-pcmatrix}. As a result, in this setting, we will always write the variables in any term of $h(x_1, \dotsc, x_K)$ in decreasing order in order to match with the ordering in Equation~\ref{eq:hqc-pcmatrix}.

We believe that the definition is best illustrated with examples. We give two.

\begin{example} \label{ex:hqc1}
    The HQC code in this example will have $K=2$, $t_1=3$, $t_2=4$, $J_1 = 4$, $J_2 = 1$, $L_1 = 8$, and $L_2=2$ and will begin with the weight matrix $W=\begin{bmatrix}1 & 1 \end{bmatrix}$. The binary parity-check matrix is given as follows:
        \begin{equation}
    H = \left[
        \begin{array}{ccc|ccc|ccc|ccc||ccc|ccc|ccc|ccc}
            0 & 0 & 0 & 1 & 0 & 0 & 0 & 0 & 0 & 0 & 0 & 0 & 0 & 0 & 0 & 0 & 0 & 0 & 0 & 0 & 0 & 0 & 1 & 0 \\
            0 & 0 & 0 & 0 & 1 & 0 & 0 & 0 & 0 & 0 & 0 & 0 & 0 & 0 & 0 & 0 & 0 & 0 & 0 & 0 & 0 & 0 & 0 & 1 \\
            0 & 0 & 0 & 0 & 0 & 1 & 0 & 0 & 0 & 0 & 0 & 0 & 0 & 0 & 0 & 0 & 0 & 0 & 0 & 0 & 0 & 1 & 0 & 0 \\ \hline
            0 & 0 & 0 & 0 & 0 & 0 & 1 & 0 & 0 & 0 & 0 & 0 & 0 & 1 & 0 & 0 & 0 & 0 & 0 & 0 & 0 & 0 & 0 & 0 \\
            0 & 0 & 0 & 0 & 0 & 0 & 0 & 1 & 0 & 0 & 0 & 0 & 0 & 0 & 1 & 0 & 0 & 0 & 0 & 0 & 0 & 0 & 0 & 0 \\
            0 & 0 & 0 & 0 & 0 & 0 & 0 & 0 & 1 & 0 & 0 & 0 & 1 & 0 & 0 & 0 & 0 & 0 & 0 & 0 & 0 & 0 & 0 & 0 \\ \hline
            0 & 0 & 0 & 0 & 0 & 0 & 0 & 0 & 0 & 1 & 0 & 0 & 0 & 0 & 0 & 0 & 1 & 0 & 0 & 0 & 0 & 0 & 0 & 0 \\
            0 & 0 & 0 & 0 & 0 & 0 & 0 & 0 & 0 & 0 & 1 & 0 & 0 & 0 & 0 & 0 & 0 & 1 & 0 & 0 & 0 & 0 & 0 & 0 \\
            0 & 0 & 0 & 0 & 0 & 0 & 0 & 0 & 0 & 0 & 0 & 1 & 0 & 0 & 0 & 1 & 0 & 0 & 0 & 0 & 0 & 0 & 0 & 0 \\ \hline
            1 & 0 & 0 & 0 & 0 & 0 & 0 & 0 & 0 & 0 & 0 & 0 & 0 & 0 & 0 & 0 & 0 & 0 & 0 & 1 & 0 & 0 & 0 & 0 \\
            0 & 1 & 0 & 0 & 0 & 0 & 0 & 0 & 0 & 0 & 0 & 0 & 0 & 0 & 0 & 0 & 0 & 0 & 0 & 0 & 1 & 0 & 0 & 0 \\
            0 & 0 & 1 & 0 & 0 & 0 & 0 & 0 & 0 & 0 & 0 & 0 & 0 & 0 & 0 & 0 & 0 & 0 & 1 & 0 & 0 & 0 & 0 & 0
            \end{array} \right]_{12 \times 24}
    \end{equation}
    Notice that the leftmost 12 columns of $H$ correspond to the $12 \times 12$ identity matrix shifted $9$ positions left, or $I_{12,9}$. However, the rightmost $12$ columns are clearly not a shifted identity matrix. But the matrix as a whole is quasi-cyclic with index $3$, and zooming out to the $4 \times 4$ blocks, we can see that the smaller $3 \times3$ matrices appear in circulant pattern. In the leftmost $12$ columns, these $3 \times 3$ circulants have the form $I_{3,0}$. In the rightmost $12$ columns, they have the form $I_{3,2}$. This leads to the following single variable representation:
    \begin{equation}
        H(x_1) = \left[
            \begin{array}{cccc|cccc}
                0 & 1 & 0 & 0 & 0 & 0 & 0 & x_1^2 \\
                0 & 0 & 1 & 0 & x_1^2 & 0 & 0 & 0 \\
                0 & 0 & 0 & 1 & 0 & x_1^2 & 0 & 0 \\
                1 & 0 & 0 & 0 & 0 & 0 & x_1^2 & 0 
            \end{array}
        \right],
    \end{equation}
    using the fact that $1 = x_1^0$ corresponds to $I_{3,0}$. Now, the first four columns of $H(x_1)$ have the form $I_{4,3}$ and last four columns have the form $x_1^2 I_{4,1}$. This motivates using a second variable to further simplify the representation.
    \begin{equation}
        H(x_1,x_2) = \left[ \begin{array}{cc} x_2^3 & x_2 x_1^2 \end{array}\right]
    \end{equation}
    Finally, we note that in this setting, $x_2^3$ corresponds to the Kronecker product $I_{4,3} \otimes I_{3,0}$ and $x_2 x_1^2$ corresponds to $I_{4,1} \otimes I_{3,2}$.
    \end{example}

The framework of HQC parity-check codes will appear in the characterization given in Theorem~\ref{thm:hqc-codes}. It will turn out that the codes constructed in this paper will have monomial entries in their polynomial parity-check representations, so Example~\ref{ex:hqc1} will be the most relevant to this work. However, we include the next, slightly more complicated, example with some weight $2$ entries to provide the reader a fuller understanding of the HQC framework.

\begin{example} \label{ex:hqc2}
    The HQC code in this example will have $K=3$, $t_1=7$, $t_2=3$, $t_3=2$, $J_1=6$, $L_1=12$, $J_2=2$, $L_2=4$, $J_3=1$, and $L_3=2$.

    Note that because $t_1=7$, $J_1=6$, and $L_1=12$, $H$ will be a $J_1 t_1 \times L_1 t_1 = 42 \times 84$ matrix, which is prohibitively large to show completely. Therefore we start first with the matrix $H(x_1)$ and recursively build up to $H(x_1,x_2,x_3)$.
    {\scriptsize
    \begin{equation}
        H(x_1) = \left[
        \begin{array}{ccc|ccc||ccc|ccc}
            0 & 1+x_1 & 0 & 1 & x_1^3 & 0 & 0 & 0 & 0 & x_1^3 & x_1^6 & 1+x_1^5 \\
            0 & 0 & 1+x_1 & 0 & 1 & x_1^3 & 0 & 0 & 0 & 1+x_1^5 & x_1^3 & x_1^6 \\
            1+x_1 & 0 & 0 & x_1^3 & 0 & 1 & 0 & 0 & 0 & x_1^6 & 1 +x_1^5 & x_1^3 \\
            \hline
            1 & x_1^3 & 0 & 0 & 1+x_1 & 0 & x_1^3 & x_1^6 & 1+x^5 & 0 & 0 & 0 \\
            0 & 1 & x_1^3 & 0 & 0 & 1+x_1 & 1+x_1^5 & x_1^3 & x_1^6 & 0 & 0 & 0 \\
            x_1^3 & 0 & 1 & 1+x_1 & 0 & 0 & x_1^6 & 1+x_1^5 & x_1^3 & 0 & 0 & 0 
        \end{array} \right]
    \end{equation}}
    {\small
    \begin{equation}
        {H(x_1,x_2) = \left[
        \begin{array}{cc|cc}
            (1+x_1)x_2^2 & 1+x_1^3 x_2^2 & 0 & x_1^3 + (1+x_1^5)x_2 + x_1^6 x_2^2 \\
            1 + x_1^3 x_2^2 & (1+x_1)x_2^2 & x_1^3 + (1+x_1^5) x_2 + x_1^6 x_2^2 & 0 
        \end{array} \right]}
    \end{equation}
    \begin{equation}
        {H(x_1,x_2,x_3) = \left[
        \begin{array}{cc}
            (1+x_1)x_2^2 + (1+x_1^3x_2^2)x_3 & (x_1^3 + (1+x_1^5)x_2 + x_1^6 x_2^2) x_3
        \end{array} \right]}
    \end{equation}}
    Let's take, for example, the polynomial $h_{1,1}(x_1,x_2,x_3)$ of $H(x_1,x_2,x_3)$. The term $(1+x_1)x_2^2$ is in the identity position in $H(x_1,x_2)$, and so it is not multiplied by any power of $x_3$. The term $(1+x_1^3x_2^2)$ is shifted left by one position in $H(x_1,x_2)$, and so is multiplied by $x_3$ in $h_{1,1}(x_1,x_2,x_3)$. 

    Now take, for example, the polynomial $h_{2,2}(x_1,x_2)$. In $H(x_1)$, in the second block row and second block column, the term $1+x_1$ is shifted two positions to the left, and so in $h_{2,2}(x_1,x_2)$, this term must be multiplied by $x_2^2$. Other terms in other polynomials can be derived similarly.
\end{example}

\subsection{The main construction}
Disjunct matrices, their applications, and their relationship to superimposed codes are detailed in \cite{du2000combinatorial}. Their application for MDPC and LDPC code constructions is discussed in \cite{haymaker2025parity}.  
\begin{definition} \label{def:disjunct}
    A matrix $M$ is \textit{$D$-disjunct} if, for any $D+1$ columns of $M$ with one column designated, there is a row with a $1$ in the designated column and a $0$ in each of the other $D$ columns.
\end{definition}

We now describe the Kautz-Singleton construction to create binary superimposed (later termed `qdisjunct'') codes from $q$-ary codes.

\begin{construction*}[Kautz, Singleton, \cite{kautz1964nonrandom}] \label{construction:KS}
Let $C$ be an $[n, k, d]_q$ code with codewords listed as column vectors $\textbf{c}_1,\ldots, \textbf{c}_{q^k}$. Order the elements of $\mathbb{F}_q$: $\alpha_1, \alpha_2, \ldots, \alpha_q$. A binary superimposed code $S(C)$ with $q^k$ codewords of length $nq$ 
is formed by replacing each symbol $\alpha_i$ with the column vector $\mathbf{e_i}$ of length $q$.
\end{construction*}

Next we present the main construction of binary parity-check codes using a matrix formed from the codewords of the Kautz-Singleton construction. 

\begin{construction}[McMillon, Haymaker \cite{haymaker2025parity}]\label{construction:MH}
    Suppose $q=p^r$ for a prime $p$ and $\F_q$ is the field with $q$ elements. Let $C$ be an $[n, k, d]_q$ code. Order the elements of $\mathbb{F}_q$: $\beta_1, \beta_2, \ldots, \beta_q$ by $p$-ary expansion.  List the codewords of $C$ in $p$-ary lexicographic order as the columns of an $n\times q^k$ $q$-ary matrix $M_C$. Then form a binary matrix $H_C$ by 
    replacing each symbol $\beta_i$ with the column vector $\mathbf{e}_i$, for each element in the matrix $M_C$ and each $i \in [q]$. Throughout, we will denote the code obtained from this construction by $\mathcal{C}(H_C)$, and its parameters by $[\tilde{n}, \tilde{k}, \tilde{d}]$.
\end{construction}

The disjunct value of a matrix from Construction~\ref{construction:MH} is tied to the  parameters of $C$ (see e.g. \cite{kautz1964nonrandom, du2000combinatorial}). In particular, for an $[n,k,d]_q$ code $C$ under Construction~\ref{construction:MH}, the disjunct value of $H_{C}$ is $D = \lfloor \frac{n-1}{n-d} \rfloor$. 

We summarize previous results on the structure of codes of the type $C(M)$ for $M$ $D$-disjunct and $\mathcal{C}(H_C)$ as in Construction \ref{construction:MH} in Propositions~\ref{prop:mindis} and \ref{prop:properties}. 

\begin{proposition}[\cite{haymaker2025parity}] \label{prop:mindis}
    Let $M$ be a $D$-disjunct matrix. Then the code $C(M)$ has minimum distance $d_{\text{min}}(C) \geq D+2$.
\end{proposition}

\begin{proposition}[see \cite{haymaker2025parity}] \label{prop:properties}
Let $C$ be an $[n,k,d]_q$ linear code.
\begin{enumerate}
    \item $H_C$ is an $nq\times q^k$ binary matrix. 

    \item If the code $C$ has no positions fixed at $0$, $\mathcal{C}(H_C)$ has an $(n, q^{k-1})$-regular Tanner graph representation, i.e. with variable nodes all of degree $n$ and check nodes all of degree $q^{k-1}$.\footnote{Otherwise, a code with $c_i=0$ for all $c\in C$ would have an all-1 row and $q-1$ all-0 rows in $H(C)$, corresponding to position $i$. }

    \item When $C$ is MDS with $k=2$, the Tanner graph of $H_C$ has girth 6. 

    \item For all codes $C$ with $d<n-1$ (and in particular for Reed-Solomon codes with $k \geq 3$), the Tanner graph of $H_C$ has girth $4$.
    \item For any code $C$, the density of ones in the matrix $H_C$ is $\frac{1}{q}$.
    \item If $C$ has no positions fixed at $0$,  the row weight of $H_C$ is $q^{k-1}$, a fraction of $\frac{1}{q}$ ones per row.
\end{enumerate}

\end{proposition}

Many of the codes we will use for Construction~\ref{construction:MH} in this paper will be Reed-Solomon codes. In full generality, however, no such constraint is required. We now give a small example to illustrate the construction.

\begin{example} \label{ex:construction}
We start with a $[3, 2, 2]$ Reed-Solomon code $C$ over $\mathbb{F}_4=\{0, 1, \alpha, \alpha+1\}$ with generator matrix 
\[ G = \begin{bmatrix}
    1 & 1 & 1 \\ 
    0 & 1 & \alpha 
\end{bmatrix}. \]
This code contains 16 codewords. Using the ordering $\alpha_1=0, \alpha_2=1, \alpha_3=\alpha, \alpha_4=\alpha+1$, we obtain the matrix $H_C$. The codewords $\mathbf{c}_i$ correspond to columns in $H_C$ as follows:
\[ 0, 1, \alpha, \alpha^2, x, x+1, x+\alpha, x+(\alpha+1), \alpha x, \ldots, (\alpha+1)x+\alpha, (\alpha+1)x+(\alpha+1).\]

For example, column 1 of the matrix $M_C$ is obtained from the evaluation of the polynomial $p_1(x)=0$ at the field elements $0, 1, \alpha$.
$H_C$ is obtained via the process where 
$0$ is mapped to $\mathbf{e}_1$, $1$ to $\mathbf{e}_2$, $\alpha$ to $\mathbf{e}_3$, and $\alpha+1$ to $\mathbf{e}_4$. The resulting matrix $H_C$ is given below.
\begin{equation} \label{eq:matrix}
    {H_C = \left[
        \begin{array}{cc|cc||cc|cc||cc|cc||cc|cc}
            1 & 0 & 0 & 0 & 1 & 0 & 0 & 0 & 1 & 0 & 0 & 0 & 1 & 0 & 0 & 0 \\
            0 & 1 & 0 & 0 & 0 & 1 & 0 & 0 & 0 & 1 & 0 & 0 & 0 & 1 & 0 & 0 \\ \hline
            0 & 0 & 1 & 0 & 0 & 0 & 1 & 0 & 0 & 0 & 1 & 0 & 0 & 0 & 1 & 0 \\
            0 & 0 & 0 & 1 & 0 & 0 & 0 & 1 & 0 & 0 & 0 & 1 & 0 & 0 & 0 & 1 \\
            \hline \hline
            1 & 0 & 0 & 0 & 0 & 1 & 0 & 0 & 0 & 0 & 1 & 0 & 0 & 0 & 0 & 1\\
            0 & 1 & 0 & 0 & 1 & 0 & 0 & 0 & 0 & 0 & 0 & 1 & 0 & 0 & 1 & 0 \\ \hline
            0 & 0 & 1 & 0 & 0 & 0 & 0 & 1 & 1 & 0 & 0 & 0 & 0 & 1 & 0 & 0 \\
            0 & 0 & 0 & 1 & 0 & 0 & 1 & 0 & 0 & 1 & 0 & 0 & 1 & 0 & 0 & 0 \\
            \hline \hline
            1 & 0 & 0 & 0 & 0 & 0 & 1 & 0 & 0 & 0 & 0 & 1 & 0 & 1 & 0 & 0 \\
            0 & 1 & 0 & 0 & 0 & 0 & 0 & 1 & 0 & 0 & 1 & 0 & 1 & 0 & 0 & 0 \\ \hline 
            0 & 0 & 1 & 0 & 1 & 0 & 0 & 0 & 0 & 1 & 0 & 0 & 0 & 0 & 0 & 1 \\
            0 & 0 & 0 & 1 & 0 & 1 & 0 & 0 & 1 & 0 & 0 & 0 & 0 & 0 & 1 & 0 
        \end{array} \right]_{12 \times 16}}
\end{equation}

When viewing $H_C$ as the parity-check matrix of a binary code $\mathcal{C}(H_C)$, the code has parameters $[16,8,4]$.

The matrix $H_C$ has $2\times 2$ circulant sub-blocks, which can be represented with a two-level hierarchical quasi-cyclic representation with $K=2$, $t_1=t_2=2$, $J_1=6, L_1=8, J_2=3, L_2=4$,
as 
\[ H_C(x_1, x_2) = \begin{bmatrix} 1 & 1& 1& 1\\ 
1& x_1& x_2 & x_2x_1\\
1& x_2 & x_2x_1 & x_1 \end{bmatrix}. \]
\end{example}

\subsection{Existing MDS-based parity-check codes}

The following construction \cite{chen2010two} was shown to unify many existing LDPC code constructions based on finite fields and MDS codes of dimension two, e.g., constructions in \cite{lan2007construction, song2009unified, lin2006algebraic, lan2005constructions}. These constructions differ from Construction \ref{construction:MH} in several ways: they restrict to cases when $k=2$,  the resulting code parameters differ, and the quasi-cyclic structure is different. Full details are given in Section~\ref{sec:equivalence}.

\begin{construction}[Chen, Bai, Wang, \cite{chen2010two}] \label{construction:LDPC}
Let $\F_q=\{0, \alpha^{0}, \alpha^1, \alpha^2, \ldots, \alpha^{q-2}\}$
be a field with a primitive element $\alpha$.  
Starting with an MDS code $C$ with parameters $[n, 2, n-1]_q$, $n\leq q+1$, partition the nonzero codewords of $C$ into classes $W_0, \ldots, W_q$ of size $(q-1)$ such that each class is closed under multiplication by nonzero field elements. The classes $W_0, \ldots, W_{(n-1)}$ contain the codewords of weight-$(n-1)$, with 0 in position $i$ of all codewords of $W_i$, respectively,  while the classes $W_{n}, \ldots, W_q$ contain the weight-$n$ codewords.
Form the matrix $W_{\text{MDS}}$ where each row $\mathbf{w}_i$ is a  representative of class $W_i$. Replace each element of $\F_q$ in $W_{\text{MDS}}$ with its $(q-1)$-fold matrix dispersion---that is,  $0$ is replaced by the $(q-1)\times(q-1)$ all-zeroes matrix; $\alpha^i$ is replaced by the $(q-1)\times(q-1)$ circulant matrix with $\textbf{u}_i$ as its first row, where the index on the circulant matrix positions goes from $0$ to $q-2$. The authors of \cite{chen2010two} name the resulting binary parity-check matrix $H_{\text{disp}}(W_{\text{MDS}})$, and its transpose as the parity-check matrix $H_{\text{II}}$.  
\end{construction}

By design, $H_{disp}(W_{MDS})$ and $H_{II}$ are both quasi-cyclic with index $q-1$. In particular, since $p\nmid q-1$, niether $H_{disp}(W_{MDS})$ nor $H_{II}$ are quasi-cyclic or hierarchical quasi-cyclic with index $p$. 
    
The following minimum distance bound for certain LDPC codes gives a slight improvement on the lower bound of $D+2$ in Proposition~\ref{prop:mindis} for cases when $n$ is even and $D<n$.  
\begin{proposition}[Xu, et al. \cite{xu2006construction}] \label{prop:dmin-LDPC} Let $C$ be a binary LDPC code with a parity-check matrix $H$ formed with  $\gamma$ rows and $\rho$ columns of block permutation matrices. Then $d\geq \gamma+2$ for $\gamma$ even. 
\end{proposition}

\section{Construction I code properties} \label{sec:parameters}

In this section, we consider the properties of codes built from Construction~\ref{construction:MH}.

\subsection{Hierarchical quasi-cyclic structure}

In this subsection, we will characterize the precise hierarchical quasi-cyclic structure of the codes built from Construction~\ref{construction:MH}. We begin with a much simpler fact, showing that they are quasi-cyclic with index equal to the characteristic of the underlying field. 

\begin{theorem} \label{thm:qc}
    Let $C$ be a $[n,k,n-k+1]_{p^r}$ Reed-Solomon code. Then $\mathcal{C}(H_C)$ is quasi-cyclic with index $p$. 
\end{theorem}

\begin{proof}
    Order the elements of $\F_{p^r}$ and the polynomials in $C$ as in Definition~\ref{def:pary}. We will show that, with this ordering, the matrix $H_C$ is quasi-cyclic with $p \times p$ circulant submatrices.

    Let $\alpha$ be a primitive element of $\F_{p^r}$. Let $\beta \in \F_q$ be arbitrary and let $q_0(x) = \sum_{i=0}^{k-1} a_i x^i$ be a polynomial in $C$ such that $a_0\in \F_{p^{r}}\setminus \F_p$.
    Then $q_0(\beta) = \sum_{i=0}^{k-1} a_i \beta^i = \sum_{i=0}^{r-1} \gamma_i \alpha^i$ for some $\gamma_i \in \F_p$. 

    Now let $q_c(x) = q_0(x) + c$ for $c \in \F_p$. Notice that 
    \begin{align*}
        q_c(\beta) &= q_0(\beta) + c = \sum_{i=0}^{r-1} \gamma_i \alpha^i + c \\
        &=  \sum_{i=1}^{r-1} \gamma_i \alpha_i + (\gamma_0+c),
    \end{align*}
    where the addition $\gamma_0 + c$ is taken modulo $p$. In particular, because of our choice in polynomial and field element ordering, this shows that the polynomial immediately after $q_0$ (modulo $p$) evaluates to the element immediately after $q_0(\beta)$, $q_0(\beta)+1$ (modulo $p$). Hence, $H_C$ is made up of $p \times p$ circulants.
\end{proof}

In fact, we can be more precise about the structure of the quasi-cyclic matrix by making use of the hierarchical quasi-cyclic framework introduced in \cite{wang2013hierarchical}, which was also defined and motivated in Section~\ref{sec:prelim}. We begin with a result that shows that individual $q\times q$ submatrices of $H_C$ are hierarchical quasi-cyclic.

\begin{proposition} \label{prop:hqc-matrices}
    Let $C$ be a $[n,k,n-k+1]_{p^r}$ Reed-Solomon code, let $H_C$ be the binary $nk\times q^k$ parity-check matrix as formed in Construction~\ref{construction:MH}. Assume the elements of $\F_{p^r}$ and polynomials in $C$ are ordered in the lexicographic ordering as in Construction~\ref{construction:MH}. Let $q(x)$ be some polynomial in $C$ with $0$ constant term and let $\beta \in \F_{p^r}$ be some evaluation point. Let $P$ be the $p^r \times p^r$ sub-matrix of $H_C$ corresponding to $q(\beta)$ and the evaluations of all shifts of $q(x)$ at $\beta$ by constant terms. $P$ is weight one hierarchical quasi-cyclic. In particular, writing $q(\beta)$ as its $p$-ary expansion $q(\beta) = \sum_{i=0}^{r-1} a_i \alpha^i$. $P$ is given by
    \[P = I_{p,a_{r-1}} \otimes I_{p,a_{r-2}} \otimes \dots \otimes I_{p,a_0}\]
    where $\otimes$ is the Kronecker product.
\end{proposition}

\begin{proof}
    $C$ has $p^{rk}$ polynomials, which we assume are ordered in lexicographic ordering using some primitive element $\alpha \in \F_{p^r}$. Thus $\alpha$ is a root of a primitive polynomial over $\mathbb{F}_p$ of degree $r$, and each element of $\mathbb{F}_{p^r}$ can be written as $\sum_{i=0}^{r-1} a_i\alpha^i$, for some $a_i\in \mathbb{F}_p$. Let $q(x)$ be some polynomial in $C$ with zero constant term and let $\beta \in \F_{p^r}$ be an evaluation point. The shifts of $q(x)$ by all constants in $\F_{p^r}$ gives $p^r$ total polynomials in lexicographic ordering. The corresponding rows are labeled in lexicographic ordering by the elements of $\F_{p^r}$, so the matrix corresponding to $q(\beta)$, which we will call $P$, is a $p^r \times p^r$ binary matrix.

    We can write $q(\beta)$ in its $p$-ary expansion as $q(\beta)= \sum_{i=0}^{r-1} a_i \alpha^i$ for some $a_i \in \F_p$. Because the nonzero entries of $P$ are coming from evaluations, it is clear that each column has weight $1$. To see that each row also has weight $1$, let $\gamma = \sum_{i=0}^{r-1} b_i \alpha^i$ and $\gamma' = \sum_{i=0}^{r-1} b_i' \alpha^i$ be two constant terms and consider $q(\beta) + \gamma = q(\beta) + \gamma'$. This gives 
    \[\sum_{i=0}^{r-1} (a_i + b_i)\alpha^i = \sum_{i=0}^{r-1} (a_i + b_i')\alpha^i\]
    which implies each $b_i = b_i'$ and so $\gamma = \gamma'$, proving that the row weights of $P$ are also $1$. 

    Let $s \in \{0, \dotsc, r-1\}$. Consider the shift of $q(\beta)$ by $\alpha^s$, 
    \[q(\beta) + \alpha^s =  \left( \sum_{i \in \{0, \dotsc, r-1\} \setminus \{s\}}^{r-1} a_i \alpha^i \right) + (a_s+1)\alpha^s\]
    where the sum $a_s+1$ is taken modulo $p$, and so for any $s$, repeating this sum accounts for exactly $p$ distinct polynomials $q(x) + \alpha^s$ and $p$ distinct outputs $q(\beta) + \alpha^s$. Taking $s=0$, we obtain the result in Theorem~\ref{thm:qc}, that $H_C$ is quasi-cyclic with index $p$.

    For any other $s$, within $P$, adding an $\alpha^s$ both corresponds to a polynomial that is $p^s$ positions forward in the lexicographic ordering on polynomials and an output value that is $p^s$ positions forward in the lexicographic ordering on field elements of $\F_{p^r}$ modulo $p^r$. Hence, adding each $\alpha^s$ gives exactly $p$ shifts both by $p^s$ rows and $p^s$ columns modulo $p^{s+1}$ within the $p^{s+1} \times p^{s+1}$ sub-block. These shifts occur $p$ times within $P$ before returning to their initial position. The resulting structure of the matrix then depends fully on the initial $p$-ary expansion $q(\beta)=\sum_{i=1}^{r-1} a_i\alpha^i$.

    More precisely, these shifts live within a sub-block of size $p^{s+1} \times p^{s+1}$ within $P$. As each row and column of $P$ has weight $1$, we see that each $p^{s+1} \times p^{s+1}$ sub-block contains all of the nonzero entries in its corresponding rows and columns.

    We will use induction to show that, for a fixed $r$, for any $s \in\{0, \dotsc, r-1\}$, the nonzero $p^s \times p^s$ submatrices of $P$ are weight $1$ hierarchical quasi-cyclic matrices of the form 
    \[ I_{p,a_s} \otimes I_{p,a_{s-1}} \otimes \dots \otimes I_{p,a_0}. \]
    
    For the base case, $s=0$, and we know that $H_C$ is quasi-cyclic with index $p$ from Theorem~\ref{thm:qc}. 
    The elements of $\F_{p^r}$ are ordered such that $\{0, \dotsc, p-1\}$ occur in that order $p^{r-1}$ times within $P$. In the first such block, the first entry corresponds to evaluation of $q(\beta)$ and so the $1$ in this column appears $a_0$ positions down from the top left corner of this $p \times p$ block. The first column of any other $p \times p$ block in $P$ block can be obtained by adding some linear combination of powers of $\alpha^i$ with $i \in [r-1]$. Notice that
    \[ q(\beta) + \sum_{i=1}^{r-1} b_i \alpha^i = \sum_{i=1}^{r-1}(b_i + a_i\alpha^i) + a_0,\]
    so the $1$ also appears in the $a_0$ position. Therefore, all $p \times p$ nonzero quasi-cyclic blocks in $P$ have the form $I_{p,a_0}$.

    Assume $1 \leq t \leq r-1$ and the claim holds for all $s \in \{0, \dotsc, t-1\}$. In particular, this means that the $p^{t-1} \times p^{t-1}$ submatrices of $P$ have the form 
    \[ I_{p,a_{t-1}} \otimes \dots \otimes I_{p,a_0}.\]
    By adding $\alpha^t$ to $q_\ell(x)$, we obtain
    \[q_\ell(\beta_j) + \alpha^t =  \left( \sum_{i \in \{0, \dotsc, r-1\} \setminus \{t\}}^{r-1} a_i \alpha^i \right) + (a_t+1)\alpha^t\]
    which corresponds to a shift by both $p^{t-1}$ rows and $p^{t-1}$ columns in $P$. The elements of $\F_{p^r}$ are ordered lexicographically, and so, each ordered set of $p^{t-1}$ elements all have the same coefficient on the $\alpha^t$. In particular, the nonzero entries of any of the $p^{t-1} \times p^{t-1}$ sub-blocks must occur in the block corresponding to the coefficient $a_t$, which is $a_t$ $p^{t-1} \times p^{t-1}$ blocks down in the $p^t \times p^t$ submatrix. The first column of any other $p^{t-1} \times p^{t-1}$ block in $P$ can be obtained by adding some linear combination of powers of $\alpha^i$ with $i \in \{t+1, \dotsc, r-1\}$. Notice that
    \[ q_\ell (\beta_j) + \sum_{i=t+1}^{r-1} b_i \alpha^i = \sum_{i=t+1}^{r-1}(b_i + a_i)\alpha^i + \sum_{i=0}^t a_i \alpha^i,\] 
    so all of the $p^{t-1} \times p^{t-1}$ nonzero submatrices are the same. Therefore, all $p^t \times p^t$ hierarchical quasi-cyclic blocks in $P$ have the form
    \[ I_{p,a_{t}} \otimes \dots \otimes I_{p,a_0},\]
    completing the inductive step.
\end{proof}

We can now show that the entire matrix $H_C$ is hierarchical quasi-cyclic and be precise about its structure.

\begin{theorem} \label{thm:hqc-codes}
    Let $C$ be a $[n,k,n-k+1]_{p^r}$ Reed-Solomon code, let $H_C$ be the binary parity-check matrix as formed in Construction~\ref{construction:MH}, and let $\mathcal{C} = \mathcal{C}(H_C)$ be its parity-check code. Assume the elements of $\F_{p^r}$ and polynomials in $C$ are ordered in the lexicographic ordering as in Construction~\ref{construction:MH}. The $p^{kr}$ polynomials in $C$ can be partitioned into $p^{r(k-1)}$ parts, each with representative a polynomial with constant term $0$. Let $q_1(x), \dotsc, q_{r(k-1)}(x)$ be these polynomials, ordered in lexicographic order, with $\alpha \in \F_{p^r}$ a fixed primitive element. For each $\beta_j$ with $j \in [n]$ and each $q_\ell(x)$ with $\ell \in [r(k-1)]$, we can write $q_\ell(\beta_j) = \sum_{i=0}^{r-1} a_i^{(j,\ell)} \alpha^i$ for some $a_i \in \F_p$. With this setup, the parity-check matrix $H_C$ can be represented as a weight $1$ level-$r$ hierarchical quasi-cyclic matrix $H(x_1, x_2, \dotsc, x_{r})$ with $J_i = np^{r-i}$ and $L_i = p^{rk-i}$ for $i \in [r]$, and in $H(x_1, \dotsc, x_{r})$, 
    \begin{equation} \label{eq:polynomial}
        h_{j,\ell}(x_1, \dotsc, x_r) = \prod_{i=0}^{r-1} x_{i+1} ^{a_i^{(j,\ell)}}.
    \end{equation}
    Equivalently, in $H_C$, the $p^r \times p^r$ block matrix $H_{j,\ell}$ is given by the Kronecker product
    \begin{equation} \label{eq:kronecker}
        H_{j,\ell} = I_{p,a_{r-1}^{(j,\ell)}} \otimes I_{p,a_{r-2}^{(j,\ell)}} \otimes \dotsm \otimes I_{p,a_0^{(j,\ell)}}.
    \end{equation}
\end{theorem}

\begin{proof}
    $C$ has $p^{rk}$ polynomials, which we assume are ordered in lexicographic ordering using some primitive element $\alpha \in \F_{p^r}$. Thus $\alpha$ is a root of a primitive polynomial over $\mathbb{F}_p$ of degree $r$, and each element of $\mathbb{F}_{p^r}$ can be written as $\sum_{i=0}^{r-1} a_i\alpha^i$, for some $a_i\in \mathbb{F}_p$. Only considering the set of polynomials with a $0$ constant term, we see that there are $p^{rk}/p^r = p^{rk-r}$ such polynomials. We label these polynomials $q_1(x), \dotsc, q_{rk-r}(x)$. We think of each $q_\ell(x)$ as representing an equivalence class $[q_\ell(x)]$ made up of $q_\ell(x)$ and all its translates by some constant in $\F_{p^r}$. These classes clearly partition the polynomials in $C$, i.e., there are $p^{rk-r}$ column blocks of polynomials, each in lexicographic ordering with size $p^r$. 

    Each pair $\beta_j$, $q_\ell(x)$ corresponds to a $p^r \times p^r$ binary matrix, which we call $H_{j,\ell}$. We now consider $q_\ell(\beta_j)$, which can be written uniquely as some linear combination of powers of $\alpha$, $q_\ell(\beta_j) = \sum_{i=0}^{r-1} a_i^{(j,\ell)}\alpha^i$ for some $a_i^{(j,\ell)} \in \F_p$. By Proposition~\ref{prop:hqc-matrices}, we obtain Equation~\ref{eq:kronecker}. Using the representation for HQC codes as given in Definition~\ref{def:hqc}, this is equivalent to the condition in Equation~\ref{eq:polynomial}.

    Each matrix $H_{j,\ell}$ is HQC of level $r$. Further, it is clear from the representation in Equation~\ref{eq:kronecker} that within each $H_{j,\ell}$, level $i$ for $i \in [r]$ has index $\tilde{J_i} = \tilde{L_i} = p^{r-i}$. Because these are all the same, this idea extends to the full matrix $H(x_1, \dotsc, x_r)$, but we need to multiply all of these by the number of $H_{j,\ell}$ matrices in each row and column. There are $n$ block rows for each evaluation point and $p^{rk-r}$ polynomials, so we obtain $J_i = n p^{r-i}$ and $L_i = p^{rk-i}$ for $i \in [r]$. 
    This completes the proof of all items in the statement.
\end{proof}

We give three examples of codes constructed in this method from Reed-Solomon codes. We have provided visual examples of hierarchical quasi-cyclic matrices in each example. Due to their size, these visuals are in the Appendix.

\begin{example} \label{ex:gf25}
    We give an example over $\F_{25}$. Let $\alpha$ be a primitive element of $\F_{25}$ and let $C$ be a $[9,2,8]_{5^2}$ Reed-Solomon code with evaluation points $\beta_1=0, \beta_2=1$, $\beta_3=2$, $\beta_4=3$, $\beta_5=4$, $\beta_6=\alpha$, $\beta_7=2\alpha$ $\beta_8=3\alpha$, and $\beta_9=4\alpha$. We can order the elements of $\F_{25}$ as follows (reading across rows first, then down columns):

    \[ \begin{array}{c c c c c}
        0 & 1 & 2 & 3 & 4 \\
        \alpha, & \alpha+1, & \alpha+2, & \alpha+3, & \alpha+4, \\
        2\alpha, & 2\alpha+1, & 2\alpha+2, & 2\alpha+3, & 2\alpha+4, \\
        3\alpha, & 3\alpha+1, & 3\alpha+2, & 3\alpha+3, & 3\alpha+4, \\
        4\alpha, & 4\alpha+1, & 4\alpha+2, & 4\alpha+3, & 4\alpha+4.
    \end{array}\]

    The polynomial elements of $C$ can be ordered similarly, with the additional constraint that $1 < x$. For illustration purposes, consider the polynomial $q(x)=(3\alpha+4)x$ and the evaluation point $q_6=\alpha$. Then the $25 \times 25$ submatrix corresponding to the entry $q(\beta_6)=q(\alpha)=2\alpha+4$ can be visualized as shown in Figure~\ref{fig:gf25} (see Appendix). The information contained in this submatrix can be expressed significantly more compactly as $x_1^2 x_0^4$. Indeed, while it would be cumbersome to write out the entire $225 \times 625$ binary matrix $H_C$, we can write it much more compactly demonstrated as the $9 \times 25$ block matrix $H(x_0,x_1)$, given below.

\noindent\scalebox{0.66}{
\small
$\begin{bmatrix}
1 & 1 & 1 & 1 & 1 & 1 & 1 & 1 & 1 & 1 & 1 & 1 & 1 & 1 & 1 & 1 & 1 & 1 & 1 & 1 & 1 & 1 & 1 & 1 & 1 \\ 
1 & x_0                    & x_0^2 & x_0^3 & x_0^4 & x_1& x_1x_0 & x_1x_0^2                    & x_1x_0^3                    & x_1x_0^4                    & x_1^2                        & x_1^2x_0                    & x_1^2x_0^2 & x_1^2x_0^3 & x_1^2x_0^4 & x_1^3                        & x_1^3x_0                    & x_1^3x_0^2 & x_1^3x_0^3 & x_1^3x_0^4 & x_1^4                        & x_1^4x_0                    & x_1^4x_0^2 & x_1^4x_0^3 & x_1^4x_0^4 \\ 
1 & x_0^2 & x_0^4 & x_0                    & x_0^3 & x_1^2                        & x_1^2x_0^2 & x_1^2x_0^4 & x_1^2x_0                    & x_1^2x_0^3 & x_1^4                        & x_1^4x_0^2 & x_1^4x_0^4 & x_1^4x_0                    & x_1^4x_0^3 & x_1& x_1x_0^2                    & x_1x_0^4 & x_1x_0 & x_1x_0^3                    & x_1^3                        & x_1^3x_0^2 & x_1^3x_0^4 & x_1^3x_0                    & x_1^3x_0^3 \\ 
1 & x_0^3 & x_0                    & x_0^4 & x_0^2 & x_1^3                        & x_1^3x_0^3 & x_1^3x_0                    & x_1^3x_0^4 & x_1^3x_0^2 & x_1& x_1x_0^3                    & x_1x_0 & x_1x_0^4                    & x_1x_0^2                    & x_1^4                        & x_1^4x_0^3 & x_1^4x_0                    & x_1^4x_0^4 & x_1^4x_0^2 & x_1^2                        & x_1^2x_0^3 & x_1^2x_0                    & x_1^2x_0^4 & x_1^2x_0^2 \\ 
1 & x_0^4 & x_0^3 & x_0^2 & x_0                    & x_1^4                        & x_1^4x_0^4 & x_1^4x_0^3 & x_1^4x_0^2 & x_1^4x_0                    & x_1^3                        & x_1^3x_0^4 & x_1^3x_0^3 & x_1^3x_0^2 & x_1^3x_0                    & x_1^2                        & x_1^2x_0^4 & x_1^2x_0^3 & x_1^2x_0^2 & x_1^2x_0                    & x_1& x_1x_0^4                    & x_1x_0^3                    & x_1x_0^2                    & x_1x_0 \\ 
1 & x_1                    & x_1^2 & x_1^3 & x_1^4 & x_1^4x_0^4 & x_0^4                        & x_1x_0^4                    & x_1^2x_0^4 & x_1^3x_0^4 & x_1^3x_0^3 & x_1^4x_0^3 & x_0^3                        & x_1x_0^3                    & x_1^2x_0^3 & x_1^2x_0^2 & x_1^3x_0^2 & x_1^4x_0^2 & x_0^2                        & x_1x_0^2                    & x_1x_0 & x_1^2x_0                    & x_1^3x_0                    & x_1^4x_0                    & x_0\\ 
1 & x_1^2 & x_1^4 & x_1                    & x_1^3 & x_1^3x_0^3 & x_0^3                        & x_1^2x_0^3 & x_1^4x_0^3 & x_1x_0^3                    & x_1x_0 & x_1^3x_0                    & x_0& x_1^2x_0                    & x_1^4x_0                    & x_1^4x_0^4 & x_1x_0^4                    & x_1^3x_0^4 & x_0^4                        & x_1^2x_0^4 & x_1^2x_0^2 & x_1^4x_0^2 & x_1x_0^2                    & x_1^3x_0^2 & x_0^2                        \\ 
1 & x_1^3 & x_1                    & x_1^4 & x_1^2 & x_1^2x_0^2 & x_0^2                        & x_1^3x_0^2 & x_1x_0^2                    & x_1^4x_0^2 & x_1^4x_0^4 & x_1^2x_0^4 & x_0^4                        & x_1^3x_0^4 & x_1x_0^4                    & x_1x_0 & x_1^4x_0                    & x_1^2x_0                    & x_0& x_1^3x_0                    & x_1^3x_0^3 & x_1x_0^3                    & x_1^4x_0^3 & x_1^2x_0^3 & x_0^3                        \\ 
1 & x_1^4 & x_1^3 & x_1^2 & x_1                    & x_1x_0  & x_0& x_1^4x_0                    & x_1^3x_0                    & x_1^2x_0                    & x_1^2x_0^2 & x_1x_0^2                    & x_0^2                        & x_1^4x_0^2 & x_1^3x_0^2 & x_1^3x_0^3 & x_1^2x_0^3 & x_1x_0^3                    & x_0^3                        & x_1^4x_0^3 & x_1^4x_0^4 & x_1^3x_0^4 & x_1^2x_0^4 & x_1x_0^4                    & x_0^4         
\end{bmatrix} $
}
    
\end{example}

\begin{example} \label{ex:gf27}
    We give an example over $\F_{27}$. Let $\alpha$ be a primitive element of $\F_{27}$ and let $C$ be a $[4,3,2]_{3^3}$ Reed-Solomon code with evaluation points $\beta_1=1$, $\beta_2=2\alpha$, $\beta_3=\alpha^2+\alpha+2$, $\beta_4=2\alpha^2+1$. We can order the elements of $\F_{27}$ as follows (reading across rows first, then down columns):
    {\footnotesize
    \[
    \begin{array}{ccccccccc}
        0, & 1, & 2, & \alpha, & \alpha+1, & \alpha+2, & 2\alpha, & 2\alpha+1, & 2\alpha+2, \\
        \alpha^2, & \alpha^2+1, & \alpha^2+2, & \alpha^2+\alpha, & \alpha^2+\alpha+1, & \alpha^2+\alpha+2, & \alpha^2+2\alpha, & \alpha^2+2\alpha+1, & \alpha^2+2\alpha+2 \\
        2\alpha^2, & 2 \alpha^2+1, & 2\alpha^2+2, & 2\alpha^2+\alpha, & 2\alpha^2+\alpha+1, & 2 \alpha^2+\alpha+2, & 2 \alpha^2+2\alpha, & 2 \alpha^2+2\alpha+1, & 2\alpha^2+2\alpha+2.
    \end{array}
    \]}
    
    The polynomial elements of $C$ can be ordered similarly, with the additional constraints that $1<x<x^2$. For illustration purposes, consider the polynomial $q(x) = \alpha x$ and the evaluation point $\beta_3=\alpha^2+\alpha+2$. Then the $27 \times 27$ submatrix corresponding to the entry $q(\beta_3)=q(\alpha^2+\alpha+2)=\alpha^2+2$ can be visualized as shown in Figure~\ref{fig:gf27} (see Appendix). The information contained in this submatrix can be expressed significantly more compactly as $x_2 x_0^2$. The full binary matrix $H_C$ in this case has size $108 \times 19683$, but the matrix $H(x_0,x_1,x_2)$ has size $4 \times 729$.
\end{example}

\begin{example}
    Finally, let $C$ be a $[25,2,24]_{5^2}$ Reed-Solomon code. This means $C$ has all $25$ evaluation points in $\F_{5^2}$. Using the usual lexicographic ordering, the polynomial parity-check matrix $H(x_0,x_1)$ has dimensions $25 \times 25$. The binary parity-check matrix $H_C$ has dimensions $625 \times 625$. The matrix $H(x_0,x_1)$ is given in Figure~\ref{fig:gf25-matrix} (see Appendix).
\end{example}

The structural setup of Proposition~\ref{prop:hqc-matrices} doesn't actually require that the polynomials be from Reed-Solomon codes, only that given any polynomial in the polynomial evaluation code, all possible constant shifts are also in the code. We defined this property to be the field partition property in Definition~\ref{def:fieldpartition}. In fact, any polynomial evaluation code with this property yields a parity-check matrix with hierarchical quasi-cyclic structure.

\begin{corollary}
    Suppose $C$ is a linear $[n, k, d]_q$ polynomial evaluation code with the field partition property over $\mathbb{F}_{p^r}$, $q=p^r$. Then the code $H_C$ formed in Construction \ref{construction:MH} is hierarchical quasi-cyclic with  $r$ levels, and $J_i = np^{r-i}$ and $L_i = p^{rk-i}$. 
\end{corollary}

\begin{proof}
    Applying the proof techniques of Proposition \ref{prop:hqc-matrices} and Theorem \ref{thm:hqc-codes} to $H_C$, the result follows.  
\end{proof}

\begin{remark}
    For example, applying Construction \ref{construction:MH} to a $p^r$-ary Reed-Muller code results in an HQC binary code defined by $H_C$. Further, many nonlinear subcodes of Reed-Solomon and Reed-Muller codes have the field partition property 
    , and thus yield additional algebraic constructions of HQC binary codes under Construction \ref{construction:MH}.  
\end{remark}

\begin{example} \label{ex:RM422}
    Let $q=4, m=2, \rho=2$, and consider the Reed-Muller code $C=RM(4, 2, 2)$. The parameters of $C$ are $[16, 6, 8]_4$. The binary matrix $H_C$ is a $64\times4096$ HQC matrix with two levels and polynomial representation $H(x_0, x_1)$ of size $16\times 1024$. 
\end{example}

\subsection{Code parameters} 

In this section, we explore code parameters for $\mathcal{C}(H_C)$. We begin with some results that hold for $q$-ary codes which are not necessarily Reed-Solomon codes. In the case of a trivial code, we can fully characterize the resulting code parameters.

\begin{proposition} If $C$ is a $[2, 2, 1]_q$ trivial code ($C=\F_q^2$) for any $q$, then $\mathcal{C}(H_{C})$  has parameters $[q^2, q^2-2q+1, 4]$. If $C$ is a $[2, 1, 2]_q$  code for any $q$ then $\mathcal{C}(H_{C})=\{\bf{0}\}$.
\end{proposition}

\begin{proof} Suppose $C$ is a $[2, 2, 1]_q$ code. 
    The matrix $H_{C}$ is a binary $2q\times q^2$ matrix, so the code length is $q^2$ and the dimension is at least $q^2-2q$. Notice that $\text{rank}(H_C)\leq 2q-1$, since the sum of rows $1$ through $q$ is the all-ones row vector, as is the sum of rows $q+1$ through $2q$. We show that the rank of $H_C$ equals $2q-1$.  The matrix $H_C$ can be arranged so that it has the form: 
    \[ H_C=\begin{pmatrix} B_{11} & B_{12} & B_{13}& \cdots & B_{1q} \\ 
                        B_{21} & B_{22} & B_{23}& \cdots & B_{2q} 
    \end{pmatrix},  \]
    where matrices $B_{11}, \ldots, B_{1q}, B_{21}$ are $q\times q$ identity matrices and $B_{ij}$ for $i=2, j=2,\ldots, q$ are as follows: 
    \[ B_{22} = \begin{pmatrix}
        \mathbf{u}_{q} \\ \mathbf{u}_1 \\ \mathbf{u}_2 \\ \vdots \\\mathbf{u}_{q-2} \\ \mathbf{u}_{q-1} 
    \end{pmatrix}, B_{23} =\begin{pmatrix}
        \mathbf{u}_{q-1} \\ \mathbf{u}_q \\ \mathbf{u}_1 \\ \vdots \\\mathbf{u}_{q-3}\\ \mathbf{u}_{q-2} 
    \end{pmatrix}, \cdots, B_{2q} = \begin{pmatrix}
        \mathbf{u}_{2} \\ \mathbf{u}_3 \\ \mathbf{u}_4 \\ \vdots \\\mathbf{u}_q\\ \mathbf{u}_{1} 
    \end{pmatrix}. \]

Let the rows of $H_C$ be denoted $\mathbf{r}_i$, for $i\in [2q]$, and suppose 
\begin{equation*}
   \sum_{i\in I} \mathbf{r}_i=\mathbf{0}\label{eq:sum-rows}
\end{equation*}
for some set $\varnothing \neq I\subseteq [2q]$.
Suppose row 1 of $H_C$, $\mathbf{r}_1=(\mathbf{u}_1, \mathbf{u}_1, \ldots, \mathbf{u}_1)$, is in $I$. Because of the structure of $B_{2j}$ for $j=1, \ldots, q$, and the location of $\mathbf{u}_1$ in each block, this implies that rows $q+1, q+2, \ldots, 2q$ are in $I$. An analogous implication follows if any other row from 1 to $q$ is in $I$. Suppose next that row $q+1$, $\mathbf{r}_{q+1}=(\mathbf{u}_1, \mathbf{u}_q, \mathbf{u}_{q-1}, \ldots, \mathbf{u}_2)$, is in $I$. In order for the rows indexed by $I$ to sum to $\bf{0}$, 
rows 1, $q$, $q-1$, 
$\ldots, 2$ must also be included in $I$. An analogous result holds for any other row indexed by $q+2, \ldots, 2q$. 
Thus $I=[2q]$ and any linear dependence involves $2q$ rows; therefore $\text{rank}(H_C)=2q-1$.  

    By Proposition~\ref{prop:dmin-LDPC}, $\tilde{d}\geq 4$. Consider the following columns $\mathbf{c}_1, \mathbf{c}_2, \mathbf{c}_3, \mathbf{c}_4$ of $H_C$: 
    \[\mathbf{c}_1=\begin{pmatrix}\mathbf{e}_1\\ \mathbf{e}_1\end{pmatrix}, \mathbf{c}_2=\begin{pmatrix}\mathbf{e}_2\\ \mathbf{e}_2\end{pmatrix}, \mathbf{c}_3=\begin{pmatrix}\mathbf{e}_1\\ \mathbf{e}_2\end{pmatrix}, \mathbf{c}_4=\begin{pmatrix}\mathbf{e}_2\\ \mathbf{e}_1\end{pmatrix}, \]
    These columns satisfy 
    $\mathbf{c}_1 + \mathbf{c}_2+\mathbf{c}_3 +\mathbf{c}_4\equiv \mathbf{0},$ 
    and thus $\tilde{d}=4$. 

    Finally, suppose $C$ is a $[2, 1, 2]_q$  code. Then the rank of $H_C$ is $q$, so $\C(H_C)=\{\bf{0}\}$. 
\end{proof}

Next we obtain an upper bound on the rank of $H_C$, and thus 
a lower bound on the dimension of the code with parity-check matrix $H_C$. An analogous bound on the code dimension was given in \cite{fossorier2004quasicyclic} (page 1) for parity-check codes with circulant permutation submatrix blocks. 

\begin{proposition} \label{prop:dim-lower-bd}
    Let $C$ be an $[n, k, d]_q$ linear code. The binary code with parity-check matrix $H_C$ has dimension at least $q^k - (nq-(n-1))$.
\end{proposition}

\begin{proof}
    The rank of $H_C$ is bounded above by $nq-(n-1)$ since the sum of each block of $q$ rows is the all-ones row vector. Thus the rank of the code with parity-check matrix $H_C$ is at least $q^k-(nq-(n-1))$. 
\end{proof}

The minimum distance bounds in Propositions \ref{prop:mindis} (with $\gamma=p^{r-1}n$)  and \ref{prop:dmin-LDPC} (with $D=\lfloor\frac{n-1}{n-d}\rfloor$) can be applied to give bounds on $\tilde{d}$. In general, these bounds are not tight (see, e.g., Table \ref{tab:parameters}). 

\begin{corollary} \label{cor:min-dist}
    Let $C$ be an $[n, k, d]_{p^r}$ linear code. The binary code with parity-check matrix $H_C$ has minimum distance satisfying: 
    \[ \tilde{d}\geq \begin{cases} p^{r-1}n+2 &\text{ for } p^{r-1}n  \text{ even} \\ 
    \lfloor\frac{n-1}{n-d}\rfloor+2 & \text{ for any } n, d.
    \end{cases} \]
\end{corollary}

\begin{example}
    The parameters of $\mathcal{C}(H_C)$ from  Example~\ref{ex:RM422} are $[4096, \tilde{k}, \tilde{d}]_2$, where $\tilde{k}\geq 4047$, $\tilde{d}\geq 34$, using Corollary \ref{cor:min-dist}.
\end{example}

When working with even characteristic, we can give an exact minimum distance in some cases.

\begin{theorem} \label{thm:even} Let $q=2^r$,  and consider an $[n, 2, n-1]$ Reed-Solomon code $C$ over $\mathbb{F}_q$, with $2< n<q$ and with evaluation points in $\F_q^*$. Then the code $\C(H_C)$ has a codeword of weight $q$ and thus $\tilde{d}\leq q$. 
\end{theorem}

\begin{proof} Since we have the lower bound $\tilde{d}\geq D+2=n+1$, we only need to demonstrate a codeword of weight $q=n+1$ in the nullspace of $H_C$. 
Because any pair of  $[n, 2, n-1]$ RS codes  over $\mathbb{F}_q$ are permutation equivalent,  we assume without loss of generality that $C$ has  evaluation points from  $\F_q^*$. 
Consider the set of polynomials   
$S= \{\beta(x-\beta)\mid \beta\in \F_q \}$.
Define the multiset $E_S=\{(i,j): p(i)=j \text{ for some } p\in S\}\subseteq \F_q^*\times \F_q$.  
    Notice that $(i,0)$ appears twice in $E_S$ for each $i\in \F_q^*$ since the zero polynomial has all zeroes and each of the other $q-1$ nonzero polynomials has a unique root. 
Let $(a, b)\in E_S$.
Consider the polynomial $\beta(x-\beta)$, for $\beta\neq 0$. Then 
\[ \beta(a-\beta)=b, \text{ if and only if } \beta^2+\beta a+b=0 \text{ in } \F_q.\]

For $a\neq 0$, the quadratic in variable $\beta$ has either 0 solutions or 2 solutions in terms of $a, b$ (see, e.g., \cite{cox2012galois}). Thus for every pair $(a,b)$, it must appear either zero or two times in $E_S$. This implies the columns of $H_C$ indexed by $S$ have an even number of 1's in each row. 
Thus the set $S$ yields a codeword in the nullspace of $H_{C}$. 
\end{proof}

\begin{corollary} \label{cor:even}
    Let $q=2^r$, and consider an $[q-1, 2, q-2]$ RS code $C$ over $\mathbb{F}_q$, with evaluation points in $\F_q^*$. Then the code $\C(H_C)$ has parameters: 
    \[ [q^2, \tilde{k}\geq 2q-2, q]. \]
\end{corollary}

\begin{proof}
    The lower bound on $\tilde{k}$ is an application of Proposition~\ref{prop:dim-lower-bd}. Since we have the lower bound $\tilde{d}\geq D+2=q$ from Proposition~\ref{prop:mindis}, Theorem~\ref{thm:even} demonstrates a codeword of weight $q$ in the nullspace of $H_C$, and thus $\tilde{d}=q$.
\end{proof}

\begin{example}
    Consider the $[3,2,2]_4$ Reed-Solomon code $C$ as given in Example~\ref{ex:construction}. The bound in Corollary~\ref{cor:even} gives that $\tilde{k} \geq 6$, while $\tilde{k}=8$ (see Table \ref{tab:parameters}). In general we observe that the dimension bound is not tight in the case of even $q$.
\end{example}

\section{Equivalence and non-equivalence} 
\label{sec:equivalence}

We begin this section with a discussion of the non-equivalence of the codes from Constructions \ref{construction:MH} and \ref{construction:LDPC}, with an explicit correspondence in the cases when the starting codes are related. While some of these codes share a common starting point of dimension two MDS $q$-ary codes, not only are the the final parity-check matrices  different, the codes themselves are not equivalent. We formalize the differences and demonstrate a process to go from one construction to the other (in the specific case when $k=2$) via column vector substitutions and other modifications.

\begin{proposition} 
Constructions~\ref{construction:MH} and \ref{construction:LDPC} produce nonisomorphic codes with different parameters and different quasi-cyclic indices. 
However, using an $[n, 2, n-1]_q$ ($n\leq q$)  Reed-Solomon code $C_{\text{RS}}$ as the starting point,  the equivalence classes of the codes $\C(H_{C_{\text{RS}}})$ from Construction~\ref{construction:MH}  are in one-to-one correspondence with the equivalence classes of the codes $\C(H_{\text{II}})$ from Construction~\ref{construction:LDPC}.   
\end{proposition}

\begin{proof} Suppose $q=p^r$, where $p$ is a prime. The code constructions are  distinct since $H_{C_{\text{RS}}}$ is an $nq\times q^2$ matrix with row weight $q$, column weight $n$, and quasi-cyclic index $p$, as well as HQC structure with $r$ levels. The parity-check matrix $H_{\text{II}}$ is an $n(q-1)\times (q^2-1)$ matrix with row weight $q$, column weight either $n$ or $(n-1)$, and quasi-cyclic index $(q-1)$ (without HQC structure). Therefore these code families are not isomorphic or permutation equivalent.

We now demonstrate a correspondence between the code families, for the special case of $k=2$ in Construction \ref{construction:MH}. Let $\F_q=\{0, 1, \alpha, \ldots, \alpha^{q-2}\}$ be a finite field. Start by ordering the polynomials that determine the codewords of $C$ as follows: 
 $\mathcal{P}_{<2}(x)=$\begin{align*} \{&x,\alpha x, \ldots, \alpha^{q-2}x, \\ 
&(x-1), \alpha(x-1), \ldots, \alpha^{q-2}(x-1), \\
&(x-\alpha), \alpha(x-\alpha), \ldots, \alpha^{q-2}(x-\alpha), \\
&\vdots \\
&(x-\alpha^{q-2}), \alpha(x-\alpha^{q-2}), \ldots, \alpha^{q-2}(x-\alpha^{q-2}), \\
&1, \alpha, \ldots, \alpha^{q-2},\\
&0\}.\end{align*} 
The first polynomial of  row $i$ for $i=0, \ldots, q$ yields a codeword that is a representative of its class $W_i$ from Construction~\ref{construction:LDPC}. Using the above ordering of $\F_q$ elements as the evaluation point ordering for Reed-Solomon codewords\footnote{Notice that this is not the $p$-ary lexicographic ordering we used in previous Sections.}, we have that $(\textbf{w}_0)_0=0$ and $(\textbf{w}_0)_j=\alpha^{j-1}$, for $j=1, \ldots, q-2$. When the field element $(\textbf{w}_i)_j$ is replaced by its $(q-1)$-fold matrix dispersion, and then the matrix transpose $H_{\text{II}}=H_{\text{disp}}(W_{\text{MDS}})^T$ is considered, each column of the $(q-1)\times (q-1)$ matrix dispersion block for $(\textbf{w}_0)_j$ records the value of each polynomial in $W_1$ at field element $\alpha^{j-1}$, respectively. An analogous statement holds for each $W_i$, $i=1,\ldots, q$. To distinguish the length of column unit vectors, we will use notation $[\textbf{e}_i]_q$ for the length-$q$ $i$th standard basis vector, for example. Making the following column vector substitutions and adding a length $nq$ column for the $p(x)=0$ polynomial, yields a permutation-equivalent code to $\mathcal{C}(H_{C_{RS}})$. 
\begin{align*}
    [\textbf{0}]_{q-1}\to &[\textbf{e}_1]_q \\
    [\textbf{e}_1]_{q-1} \to & [\textbf{e}_2]_q\\
    \vdots & \\ 
    [\textbf{e}_{q-1}]_{q-1} \to & [\textbf{e}_q]_q
\end{align*}
Doing the reverse substitutions, deleting the $0$ polynomial column, and permuting the column and row blocks, would take a matrix $H_{C_{\text{RS}}}$ to a corresponding matrix $H_{\text{II}}$. 
\end{proof}

In general, the above correspondence could serve as a useful framework to potentially obtain hierarchical quasi-cyclic families of codes from other algebraic LDPC code families, but this is an open line of inquiry.

Next we address notions of code equivalence for codes from Construction~\ref{construction:MH}, given some equivalence of the starting codes.

Recall that two codes $C_1$ and $C_2$ of length $n$ are \textit{permutation equivalent} if there exists a permutation $\pi \in S_n$ such that $\pi$ acting on the coordinates of $C_1$ is $C_2$.

Two codes $C_1$ and $C_2$ of length $n$ over $\F_q$ are \textit{monomially equivalent} if, given generator matrix $G_1$ for $C_1$, there exists a monomial matrix $M$ such that $G_1 M$ is a generator matrix for $C_2$. A \textit{monomial matrix} is a matrix $M$ that can be decomposed as $DP$, where $D$ is a diagonal matrix with entries in $\F_q$ and $P$ is a permutation matrix. Note that in the case of binary codes, monomial and permutation equivalence are the same.

\begin{remark}\label{remark:mono}
   Reed-Solomon codes over $\F_q$ with common parameters $[n, 2, n-1]$, $n< q$ formed using different sets of evaluation points are monomially equivalent.
   This is a corollary of Theorem 4.15 in \cite{cabana2022cyclic}.
\end{remark}

\begin{example}
    Consider a Reed-Solomon code over $\mathbb{F}_4$ with parameters $[3, 2, 2]$. Notice that codes generated by generator matrices $G_1, G_2, G_3, G_4$ are all permutation equivalent, using a different choice of linearly independent rows for the generator matrix: 
    \[ G_1= \begin{blockarray}{cccc}
0 & 1 & \alpha \\
\begin{block}{(ccc)c}
  1& 1&1& p=1 \\
  0&1&\alpha & p=x\\
\end{block}
\end{blockarray} \]
\[
G_2= \begin{blockarray}{cccc}
1 & \alpha & \alpha^2 \\
\begin{block}{(ccc)c}
  1& 1&1& p=1 \\
  1 & \alpha & \alpha^2 & p=x\\
\end{block}
\end{blockarray}\cong \begin{blockarray}{cccc}
1 & \alpha & \alpha^2 \\
\begin{block}{(ccc)c}
  1& 1&1& p=1 \\
  0&\alpha &1 & p=\alpha^2(x+1)\\
\end{block}
\end{blockarray}
\]
\[
G_3= \begin{blockarray}{cccc}
0 & \alpha & \alpha^2 \\
\begin{block}{(ccc)c}
  1& 1&1& p=1 \\
  0 & \alpha & \alpha^2 & p=x\\
\end{block}
\end{blockarray}\cong \begin{blockarray}{cccc}
0 & \alpha & \alpha^2 \\
\begin{block}{(ccc)c}
  1& 1&1& p=1 \\
  0&\alpha &1 & p=\alpha^2x\\
\end{block}
\end{blockarray}
\]
\[
G_4= \begin{blockarray}{cccc}
0 & 1 & \alpha^2 \\
\begin{block}{(ccc)c}
  1& 1&1& p=1 \\
  0 & 1 & \alpha^2 & p=x\\
\end{block}
\end{blockarray}\cong \begin{blockarray}{cccc}
0 & 1 & \alpha^2 \\
\begin{block}{(ccc)c}
  1& 1&1& p=1 \\
  0&\alpha &1 & p=\alpha x\\
\end{block}
\end{blockarray}
\]

\end{example}

\begin{proposition}
    If $C_1$ and $C_2$ are permutation equivalent,  then $\mathcal{C}(H_{C_1})=\mathcal{C}(H_{C_2})$.

\end{proposition}

\begin{proof}
    Codes $C_1$ and $C_2$ are permutation equivalent if there is a generator matrices $G_1$ of $C_1$ such that a permutation of columns of $G_1$ generates $C_2$. Equivalently, an $n \times q^k$ array $M_{C_1}$ containing the codewords of $C_1$ as columns can be transformed into an $n \times q^k$ array $M_{C_2}$ containing the codewords of $C_2$ as columns, via a permutation of the rows of $M_{C_1}$. Expanding the $M_{C_1}$ array to an $nq\times q^k$ matrix $H_{C_1}$, we see that a permutation of the rows of $M_{C_1}$ is simply a permutation of the $n\times q^k$ submatrices corresponding to each codeword position in $C_1$. The nullspace of $H_{C_1}$ is the same as the nullspace of a block-permutation of $H_{C_1}$, thus we have $\mathcal{C}(H_{C_1})=\mathcal{C}(H_{C_2})$. 
\end{proof}

\begin{theorem} \label{thm:monequiv}
    Suppose $C_1$ and $C_2$ are monomially equivalent. Then $\mathcal{C}(H_{C_1})$ and $\mathcal{C}(H_{C_2})$ are permutation equivalent and therefore monomially equivalent.
\end{theorem}

\begin{proof}
    Let $G_1$ be a generator matrix for $C_1$. There exists a diagonal matrix $D$ and permutation matrix $P$ such that $G_1 D P$ generates $C_2$. In particular, if $\mathbf{c} \in C_1$, then $\mathbf{c} D P \in C_2$. Let $\alpha_i$ be the $i$th entry on the  diagonal of $D$ and $\pi$ be the permutation corresponding to $P$. Under this operation, $c_i \mapsto \alpha_i c_{\pi(i)}$. In $H_{C_1}$, for all $i \in [n]$, the column corresponding to $\mathbf{c}$ has a $1$ in the row corresponding to $c_i$ of block row $i$ and zeroes elsewhere in block row $i$. In $H_{C_2}$, $\mathbf{c}DP$, for all $i \in [n]$, is such that the column corresponding to $\mathbf{c}DP$ has a $1$ in the row corresponding to $\alpha_i c_i$ and zeroes elsewhere in block row $\pi(i)$. Because there is a bijection between codewords in $\mathbf{c} \in C_1$ and $\mathbf{c}DP \in C_2$, this is true for all codewords. So $H_{C_2}$ can be obtained via row and column permutations from $H_{C_1}$ and $\mathcal{C}({H_{C_1}})$ and $\mathcal{C}({H_{C_2}})$ are permutation equivalent (and hence monomially equivalent, as they are binary).
\end{proof}

Given the preceding results, one might wonder if nonbinary codes with the same parameters always yield binary parity-check codes with the same parameters, but they do not. 

\begin{example}
We give three non-monomially equivalent $[5,3,2]_3$ codes $C_1$, $C_2$, and $C_3$ with generator matrices $G_1$, $G_2$, and $G_3$, respectively.
\[ G_1=\begin{bmatrix}
    1&0&0&1&1\\
    0&1&0&1&2\\
    0&0&1&0&2 
\end{bmatrix} \quad 
G_2 = \begin{bmatrix}
    1&0&0&2&2\\
    0&1&0&2&2\\
    0&0&1&2&2     
\end{bmatrix} \quad
G_3 = \begin{bmatrix}
    1&0&0&1&1\\
    0&1&0&2&0\\
    0&0&1&0&2
\end{bmatrix}\]

The corresponding binary codes $\mathcal{C}(H_{C_1})$ and $\mathcal{C}(H_{C_3})$ both have parameters $[27,16,4]$, while $\mathcal{C}(H_{C_2})$ has parameters $[27,18,4]$. We further remark that $C_1$ has weight enumerator $4x^5+6x^4+14x^3+2x^2+1$ and $C_3$ has weight enumerator $2x^5+12x^4+8x^3+4x^2+1$, which demonstrates that codes with different weight enumerators can yield binary $H_C$ codes with the same parameters. The code $C_2$ has weight enumerator $6x^5+6x^4+8x^3+6x^2+1$. 
\end{example}

\section{Table of code parameters}
\label{sec:table}

We calculated the parameters of the codes obtained in Construction~\ref{construction:MH} for some small field sizes. In Table~\ref{tab:parameters}, the parameters $[n,k]_q$ of the starting Reed-Solomon code are given in the first column, the parameters $[\tilde{n},\tilde{k},\tilde{d}]$ obtained via Construction~\ref{construction:MH} in the second column, and the best known $d$ for a $[\tilde{n},\tilde{k}]_2$ code from Grassl's codetables.de \cite{Grassl:codetables} in the final column. In the middle column, an asterisk indicates that that code meets the best known minimum distance for its length and dimension. By Theorem~\ref{thm:monequiv} and Remark~\ref{remark:mono}, the resulting code parameters are independent of the choice of evaluation points of the starting Reed-Solomon codes.

\begin{table}[htbp]
\caption{Construction~\ref{construction:MH} code parameters.}
\begin{center}
\begin{tabular}{|c|c|c|c|}
\hline
RS Code & Construction 1 & $H_C$ Dim. & Best Known $d$ \\ \hline
$[3,2]_3$ & $[9,2,6]^*$ & $9 \times 9$ & 6 \\
$[2,2]_3$ & $[9,4,4]^*$ & $6 \times 9$ & 4 \\
$[4,2]_4$ & $[16,7,6]^*$ & $16 \times 16$ & 6 \\
$[3,2]_4$ & $[16,8,4]$ & $12 \times 16$ & 5 \\
$[2,2]_4$ & $[16,9,4]^*$ & $8 \times 16$ & 4 \\
$[5,2]_5$ & $[25,4,10]$ & $25 \times 25$ & 12 \\
$[4,2]_5$ & $[25,8,8]$ & $20 \times 25$ & 9 \\
$[3,2]_5$ & $[25,12,6]$ & $15 \times 25$ & 8 \\
$[2,2]_5$ & $[25,16,4]^*$ & $10 \times 25$ & 4 \\
$[7,2]_7$ & $[49,6,14]$ & $49 \times 49$ & 24 \\
$[6,2]_7$ & $[49,12,12]$ & $42 \times 49$ & 18-19 \\
$[5,2]_7$ & $[49,18,12]$ & $35 \times 49$ & 14-15 \\
$[4,2]_7$ & $[49,24,10]$ & $28 \times 49$ & 12 \\
$[3,2]_7$ & $[49,30,6]$ & $21 \times 49$ & 8 \\
$[2,2]_7$ & $[49,36,4]$ & $14 \times 49$ & 6 \\
$[8,2]_8$ & $[64,37,10]^*$ & $64 \times 64$ & 10-12 \\
$[7,2]_8$ & $[64,38,8]$ & $56 \times 64$ & 10-12 \\
$[6,2]_8$ & $[64,39,8]$ & $48 \times 64$ & 10-12 \\
$[5,2]_8$ & $[64,40,8]$ & $40 \times 64$ & 9-11 \\
$[4,2]_8$ & $[64,41,8]^*$ & $32 \times 64$ & 8-10 \\
$[8,2]_9$ & $[81,16,16]$ & $72 \times 81$ & 31-32 \\
$[7,2]_9$ & $[81,24,16]$ & $63 \times 81$ & 24-28 \\
$[6,2]_9$ & $[81,32,16]$ & $54 \times 81$ & 18-24 \\
$[5,2]_9$ & $[81,40,10]$ & $45 \times 81$ & 16-20 \\ 
$[10,2]_{11}$ & $[121,20,20]$ & $110 \times 121$ & 45-50 \\ 
$[12,2]_{13}$ & $[169,24,24]$ & $156 \times 169$ & 60-70 \\
\hline
\end{tabular}
\label{tab:parameters}
\end{center}
\end{table}

Because our codes are quasi-cyclic (or HQC) and those in the Code Tables are not, the codes constructed here that meet the given bounds are different codes than those in the Code Tables. In addition to being different, these codes have many desirable properties: they are quasi-cyclic (or HQC), have Tanner graph representations of girth $6$, and are constructed via an algebraic method.

\section{Conclusions}
\label{sec:conclusion}

This paper gives an algebraic approach to obtaining a family of quasi-cyclic parity-check codes obtained from Reed-Solomon codes and polynomial evaluation codes with the field partition property via the Kautz-Singleton superimposed code construction. We showed that the construction yields explicit algebraic HQC codes, which to our knowledge is the first of its kind.

We analyzed several structural properties of the resulting codes and provided some explicit parameters and bounds. In particular, when the construction is applied to Reed-Solomon codes of dimension two, the resulting parity-check codes have Tanner graph representations of girth six, a desirable property for iterative decoding. At small blocklengths, we showed that some of these codes meet the parameters of the currently best known binary codes.

More broadly, this paper illustrates a connection between classical nonbinary algebraic codes and structured hierarchical quasi-cyclic parity-check codes. This framework suggests that further algebraic constructions of hierarchical quasi-cyclic codes may be possible and that their parameters can be studied using similar techniques.

This work suggests several directions for future research. Here, we only considered Reed-Solomon and certain polynomial evaluation codes, but any $q$-ary code can be used in the Kautz-Singleton construction. It would be interesting to determine the properties of other binary parity-check codes from this construction, and to characterize when the resulting codes are quasi-cyclic.

\newpage
\appendix
\section{}

    \begin{figure}[h]
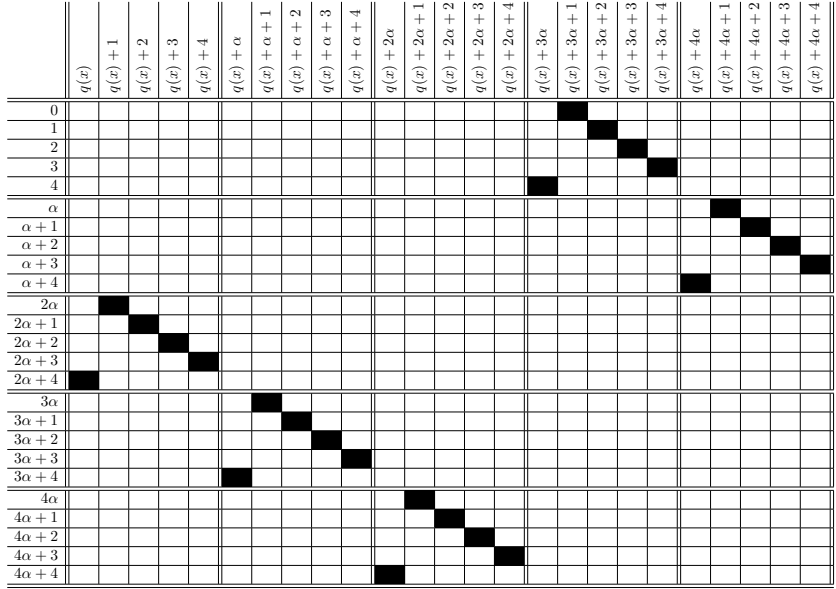

        \centering
        \scalebox{.53}{
        \begin{tabular}{r||c|c|c|c|c||c|c|c|c|c||c|c|c|c|c||c|c|c|c|c||c|c|c|c|c||}
            & \begin{sideways} $q(x)$ \end{sideways} & \begin{sideways} $q(x)+1$ \end{sideways} & \begin{sideways} $q(x)+2$ \end{sideways} & \begin{sideways} $q(x)+3$ \end{sideways} & \begin{sideways} $q(x)+4$ \end{sideways} & \begin{sideways} $q(x)+\alpha$ \end{sideways} & \begin{sideways} $q(x)+\alpha+1$ \end{sideways} & \begin{sideways} $q(x)+\alpha+2$ \end{sideways} & \begin{sideways} $q(x)+\alpha+3$ \end{sideways} & \begin{sideways} $q(x)+\alpha+4$ \end{sideways} & \begin{sideways} $q(x)+2\alpha$ \end{sideways} & \begin{sideways} $q(x)+2\alpha+1$ \end{sideways} & \begin{sideways} $q(x)+2\alpha+2$ \end{sideways} & \begin{sideways} $q(x)+2\alpha+3$ \end{sideways} & \begin{sideways} $q(x)+2\alpha+4$ \end{sideways} & \begin{sideways} $q(x)+3\alpha$ \end{sideways} & \begin{sideways} $q(x)+3\alpha+1$ \end{sideways} & \begin{sideways} $q(x)+3\alpha+2$ \end{sideways} & \begin{sideways} $q(x)+3\alpha+3$ \end{sideways} & \begin{sideways} $q(x)+3\alpha+4$ \end{sideways} & \begin{sideways} $q(x)+4\alpha$ \end{sideways} & \begin{sideways} $q(x)+4\alpha+1$ \end{sideways} & \begin{sideways} $q(x)+4\alpha+2$ \end{sideways} & \begin{sideways} $q(x)+4\alpha+3$ \end{sideways} & \begin{sideways} $q(x)+4\alpha+4$ \end{sideways} \\ \hline \hline
            $0$         & & & & & & & & & & & & & & & & & \cellcolor{black} & & & & & & & & \\ \hline
            $1$         & & & & & & & & & & & & & & & & & & \cellcolor{black} & & & & & & & \\ \hline
            $2$         & & & & & & & & & & & & & & & & & & & \cellcolor{black} & & & & & & \\ \hline
            $3$         & & & & & & & & & & & & & & & & & & & & \cellcolor{black} & & & & & \\ \hline
            $4$         & & & & & & & & & & & & & & & & \cellcolor{black} & & & & & & & & & \\ \hline \hline
            $\alpha$    & & & & & & & & & & & & & & & & & & & & & & \cellcolor{black} & & & \\ \hline
            $\alpha+1$  & & & & & & & & & & & & & & & & & & & & & & & \cellcolor{black} & & \\ \hline
            $\alpha+2$  & & & & & & & & & & & & & & & & & & & & & & & & \cellcolor{black} & \\ \hline
            $\alpha+3$  & & & & & & & & & & & & & & & & & & & & & & & & & \cellcolor{black} \\ \hline
            $\alpha+4$  & & & & & & & & & & & & & & & & & & & & & \cellcolor{black} & & & & \\ \hline \hline
            $2\alpha$   & & \cellcolor{black} & & & & & & & & & & & & & & & & & & & & & & & \\ \hline
            $2\alpha+1$ & & & \cellcolor{black} & & & & & & & & & & & & & & & & & & & & & & \\ \hline
            $2\alpha+2$ & & & & \cellcolor{black} & & & & & & & & & & & & & & & & & & & & & \\ \hline
            $2\alpha+3$ & & & & & \cellcolor{black} & & & & & & & & & & & & & & & & & & & & \\ \hline
            $2\alpha+4$ & \cellcolor{black} & & & & & & & & & & & & & & & & & & & & & & & & \\ \hline \hline
            $3\alpha$   & & & & & & & \cellcolor{black} & & & & & & & & & & & & & & & & & & \\ \hline
            $3\alpha+1$ & & & & & & & & \cellcolor{black} & & & & & & & & & & & & & & & & & \\ \hline
            $3\alpha+2$ & & & & & & & & & \cellcolor{black} & & & & & & & & & & & & & & & & \\ \hline 
            $3\alpha+3$ & & & & & & & & & & \cellcolor{black} & & & & & & & & & & & & & & & \\ \hline
            $3\alpha+4$ & & & & & & \cellcolor{black} & & & & & & & & & & & & & & & & & & & \\ \hline \hline
            $4\alpha$   & & & & & & & & & & & & \cellcolor{black} & & & & & & & & & & & & & \\ \hline
            $4\alpha+1$ & & & & & & & & & & & & & \cellcolor{black} & & & & & & & & & & & & \\ \hline
            $4\alpha+2$ & & & & & & & & & & & & & & \cellcolor{black} & & & & & & & & & & & \\ \hline
            $4\alpha+3$ & & & & & & & & & & & & & & & \cellcolor{black} & & & & & & & & & & \\ \hline
            $4\alpha+4$ & & & & & & & & & & & \cellcolor{black} & & & & & & & & & & & & & & \\ \hline \hline
        \end{tabular}}
        \caption{A visualization of the $25 \times 25$ submatrix corresponding to evaluation point $\beta=\alpha$ and polynomial $q(x)=(3 \alpha + 4)x$ as in Example~\ref{ex:gf25}. Black squares correspond to $1$s and white squares to $0$s.} \label{fig:gf25}
    \end{figure}

        \begin{figure}[h]
        \centering
        \scalebox{.53}{
        \begin{tabular}{r|||c|c|c||c|c|c||c|c|c|||c|c|c||c|c|c||c|c|c|||c|c|c||c|c|c||c|c|c|||}
            & \begin{sideways} $q(x)$ \end{sideways} & \begin{sideways} $q(x)+1$ \end{sideways} & \begin{sideways} $q(x)+2$ \end{sideways} & \begin{sideways} $q(x)+\alpha$ \end{sideways} & \begin{sideways} $q(x)+\alpha+1$ \end{sideways} & \begin{sideways} $q(x)+\alpha+2$ \end{sideways} & \begin{sideways} $q(x)+2\alpha$ \end{sideways} & \begin{sideways} $q(x)+2\alpha+1$ \end{sideways} & \begin{sideways} $q(x)+2\alpha+2$ \end{sideways} & \begin{sideways} $q(x)+\alpha^2$ \end{sideways} & \begin{sideways} $q(x)+\alpha^2+1$ \end{sideways} & \begin{sideways} $q(x)+\alpha^2+2$ \end{sideways} & \begin{sideways} $q(x)+\alpha^2+\alpha$ \end{sideways} & \begin{sideways} $q(x)+\alpha^2+\alpha+1$ \end{sideways} & \begin{sideways} $q(x)+\alpha^2+\alpha+2$ \end{sideways} & \begin{sideways} $q(x)+\alpha^2+2\alpha$ \end{sideways} & \begin{sideways} $q(x)+\alpha^2+2\alpha+1$ \end{sideways} & \begin{sideways} $q(x)+\alpha^2+2\alpha+2$ \end{sideways} & \begin{sideways} $q(x)+2\alpha^2$ \end{sideways} & \begin{sideways} $q(x)+2\alpha^2+1$ \end{sideways} & \begin{sideways} $q(x)+2\alpha^2+2$ \end{sideways} & \begin{sideways} $q(x)+2\alpha^2+\alpha$ \end{sideways} & \begin{sideways} $q(x)+2\alpha^2+\alpha+1$ \end{sideways} & \begin{sideways} $q(x)+2\alpha^2+\alpha+2$ \end{sideways} & \begin{sideways} $q(x)+2\alpha^2+2\alpha$ \end{sideways} & \begin{sideways} $q(x)+2\alpha^2+2\alpha+1$ \end{sideways} & \begin{sideways} $q(x)+2\alpha^2+2\alpha+2$ \end{sideways} \\ \hline \hline \hline
            $0$ & & & & & & & & & & & & & & & & & & & & \cellcolor{black} & & & & & & & \\ \hline
            $1$ & & & & & & & & & & & & & & & & & & & & & \cellcolor{black} & & & & & & \\ \hline
            $2$ & & & & & & & & & & & & & & & & & & & \cellcolor{black} & & & & & & & & \\ \hline \hline
            $\alpha$ & & & & & & & & & & & & & & & & & & & & & & & \cellcolor{black} & & & & \\ \hline
            $\alpha+1$ & & & & & & & & & & & & & & & & & & & & & & & & \cellcolor{black} & & & \\ \hline
            $\alpha+2$ & & & & & & & & & & & & & & & & & & & & & & \cellcolor{black} & & & & & \\ \hline \hline
            $2\alpha$ & & & & & & & & & & & & & & & & & & & & & & & & & & \cellcolor{black} & \\ \hline
            $2\alpha+1$ & & & & & & & & & & & & & & & & & & & & & & & & & & & \cellcolor{black} \\ \hline
            $2\alpha+2$ & & & & & & & & & & & & & & & & & & & & & & & & & \cellcolor{black} & & \\ \hline \hline \hline
            $\alpha^2$ & & \cellcolor{black} & & & & & & & & & & & & & & & & & & & & & & & & & \\ \hline
            $\alpha^2+1$ & & & \cellcolor{black} & & & & & & & & & & & & & & & & & & & & & & & & \\ \hline
            $\alpha^2+2$ & \cellcolor{black} & & & & & & & & & & & & & & & & & & & & & & & & & & \\ \hline \hline
            $\alpha^2+\alpha$ & & & & & \cellcolor{black} & & & & & & & & & & & & & & & & & & & & & & \\ \hline
            $\alpha^2+\alpha+1$ & & & & & & \cellcolor{black} & & & & & & & & & & & & & & & & & & & & & \\ \hline
            $\alpha^2+\alpha+2$ & & & & \cellcolor{black} & & & & & & & & & & & & & & & & & & & & & & &\\ \hline \hline
            $\alpha^2+2\alpha$ & & & & & & & & \cellcolor{black} & & & & & & & & & & & & & & & & & & & \\ \hline
            $\alpha^2+2\alpha+1$ & & & & & & & & & \cellcolor{black} & & & & & & & & & & & & & & & & & & \\ \hline
            $\alpha^2+2\alpha+2$ & & & & & & & \cellcolor{black} & & & & & & & & & & & & & & & & & & & & \\ \hline  \hline \hline
            $2\alpha^2$ & & & & & & & & & & & \cellcolor{black} & & & & & & & & & & & & & & & & \\ \hline
            $2\alpha^2+1$ & & & & & & & & & & & & \cellcolor{black} & & & & & & & & & & & & & & & \\ \hline
            $2\alpha^2+2$ & & & & & & & & & & \cellcolor{black} & & & & & & & & & & & & & & & & & \\ \hline \hline
            $2\alpha^2+\alpha$ & & & & & & & & & & & & & & \cellcolor{black} & & & & & & & & & & & & & \\ \hline
            $2\alpha^2+\alpha+1$ & & & & & & & & & & & & & & & \cellcolor{black} & & & & & & & & & & & & \\ \hline
            $2\alpha^2+\alpha+2$ & & & & & & & & & & & & & \cellcolor{black} & & & & & & & & & & & & & & \\ \hline \hline
            $2\alpha^2+2\alpha$ & & & & & & & & & & & & & & & & & \cellcolor{black} & & & & & & & & & & \\ \hline
            $2\alpha^2+2\alpha+1$ & & & & & & & & & & & & & & & & & & \cellcolor{black} & & & & & & & & & \\ \hline
            $2\alpha^2+2\alpha+2$ & & & & & & & & & & & & & & & & \cellcolor{black} & & & & & & & & & & & \\ \hline \hline \hline
        \end{tabular}}
        \caption{A visualization of the $27 \times 27$ submatrix corresponding to evaluation point $\beta = \alpha^2+\alpha+2$ and polynomial $q(x)=\alpha x$ as in Example~\ref{ex:gf27}. Black squares correspond to $1$s and white squares to $0$s.} \label{fig:gf27}
    \end{figure}

\begin{sidewaysfigure}
\centering
    
\noindent\scalebox{0.75}{
$
\begin{blockarray}{rccccccccccccccccccccccccc}
& \rotatebox{90}{0} & \rotatebox{90}{$x$} & \rotatebox{90}{$2x$} & \rotatebox{90}{$3x$} & \rotatebox{90}{$4x$} & \rotatebox{90}{$\alpha x$} & \rotatebox{90}{$(\alpha+1)x$} & \rotatebox{90}{$(\alpha+2)x$} & \rotatebox{90}{$(\alpha+3)x$} & \rotatebox{90}{$(\alpha+4)x$} & \rotatebox{90}{$2\alpha x $} & \rotatebox{90}{$(2\alpha+1)x$} & \rotatebox{90}{$(2\alpha+2)x$} & \rotatebox{90}{$(2\alpha+3)x$} & \rotatebox{90}{$(2\alpha+4)x$} & \rotatebox{90}{$3\alpha x$} & \rotatebox{90}{$(3\alpha+1)x$} & \rotatebox{90}{$(3\alpha+2)x$} & \rotatebox{90}{$(3 \alpha + 3)x$} & \rotatebox{90}{$(3\alpha+4)x$} & \rotatebox{90}{$4\alpha x$} & \rotatebox{90}{$(4\alpha+1)x$} & \rotatebox{90}{$(4\alpha+2)x$} & \rotatebox{90}{$(4\alpha+3)x$} & \rotatebox{90}{$(4\alpha+4)x$} \\
\begin{block}{r[ccccccccccccccccccccccccc]}
0 & 0 & 0 & 0 & 0 & 0 & 0 & 0 & 0 & 0 & 0 & 0 & 0 & 0 & 0 & 0 & 0 & 0 & 0 & 0 & 0 & 0 & 0 & 0 & 0 & 0 \\
x_0 & 0 & x_0 & x_0^2 & x_0^3 & x_0^4 & x_1 & x_1x_0 & x_1x_0^2 & x_1x_0^3 & x_1x_0^4 & x_1^2 & x_1^2x_0 & x_1^2x_0^2 & x_1^2x_0^3 & x_1^2x_0^4 & x_1^3 & x_1^3x_0 & x_1^3x_0^2 & x_1^3x_0^3 & x_1^3x_0^4 & x_1^4 & x_1^4x_0 & x_1^4x_0^2 & x_1^4x_0^3 & x_1^4x_0^4 \\
x_0^2 & 0 & x_0^2 & x_0^4 & x_0 & x_0^3 & x_1^2 & x_1^2x_0^2 & x_1^2x_0^4 & x_1^2x_0 & x_1^2x_0^3 & x_1^4 & x_1^4x_0^2 & x_1^4x_0^4 & x_1^4x_0 & x_1^4x_0^3 & x_1 & x_1x_0^2 & x_1x_0^4 & x_1x_0 & x_1x_0^3 & x_1^3 & x_1^3x_0^2 & x_1^3x_0^4 & x_1^3x_0 & x_1^3x_0^3 \\
x_0^3 & 0 & x_0^3 & x_0 & x_0^4 & x_0^2 & x_1^3 & x_1^3x_0^3 & x_1^3x_0 & x_1^3x_0^4 & x_1^3x_0^2 & x_1 & x_1x_0^3 & x_1x_0 & x_1x_0^4 & x_1x_0^2 & x_1^4 & x_1^4x_0^3 & x_1^4x_0 & x_1^4x_0^4 & x_1^4x_0^2 & x_1^2 & x_1^2x_0^3 & x_1^2x_0 & x_1^2x_0^4 & x_1^2x_0^2 \\
x_0^4 & 0 & x_0^4 & x_0^3 & x_0^2 & x_0 & x_1^4 & x_1^4x_0^4 & x_1^4x_0^3 & x_1^4x_0^2 & x_1^4x_0 & x_1^3 & x_1^3x_0^4 & x_1^3x_0^3 & x_1^3x_0^2 & x_1^3x_0 & x_1^2 & x_1^2x_0^4 & x_1^2x_0^3 & x_1^2x_0^2 & x_1^2x_0 & x_1 & x_1x_0^4 & x_1x_0^3 & x_1x_0^2 & x_1x_0 \\
\alpha & 0 & x_1 & x_1^2 & x_1^3 & x_1^4 & x_1x_0^3 & x_1^2x_0^3 & x_1^3x_0^3 & x_1^4x_0^3 & x_0^3 & x_1^2x_0 & x_1^3x_0 & x_1^4x_0 & x_0 & x_1x_0 & x_1^3x_0^4 & x_1^4x_0^4 & x_0^4 & x_1x_0^4 & x_1^2x_0^4 & x_1^4x_0^2 & x_0^2 & x_1x_0^2 & x_1^2x_0^2 & x_1^3x_0^2 \\
\alpha+1 & 0 & x_1x_0 & x_1^2x_0^2 & x_1^3x_0^3 & x_1^4x_0^4 & x_1^2x_0^3 & x_1^3x_0^4 & x_1^4 & x_0 & x_1x_0^2 & x_1^4x_0 & x_0^2 & x_1x_0^3 & x_1^2x_0^4 & x_1^3 & x_1x_0^4 & x_1^2 & x_1^3x_0 & x_1^4x_0^2 & x_0^3 & x_1^3x_0^2 & x_1^4x_0^3 & x_0^4 & x_1 & x_1^2x_0 \\
\alpha+2 & 0 & x_1x_0^2 & x_1^2x_0^4 & x_1^3x_0 & x_1^4x_0^3 & x_1^3x_0^3 & x_1^4 & x_0^2 & x_1x_0^4 & x_1^2x_0 & x_1x_0 & x_1^2x_0^3 & x_1^3 & x_1^4x_0^2 & x_0^4 & x_1^4x_0^4 & x_0 & x_1x_0^3 & x_1^2 & x_1^3x_0^2 & x_1^2x_0^2 & x_1^3x_0^4 & x_1^4x_0 & x_0^3 & x_1 \\
\alpha+3 & 0 & x_1x_0^3 & x_1^2x_0 & x_1^3x_0^4 & x_1^4x_0^2 & x_1^4x_0^3 & x_0 & x_1x_0^4 & x_1^2x_0^2 & x_1^3 & x_1^3x_0 & x_1^4x_0^4 & x_0^2 & x_1 & x_1^2x_0^3 & x_1^2x_0^4 & x_1^3x_0^2 & x_1^4 & x_0^3 & x_1x_0 & x_1x_0^2 & x_1^2 & x_1^3x_0^3 & x_1^4x_0 & x_0^4 \\
\alpha+4 & 0 & x_1x_0^4 & x_1^2x_0^3 & x_1^3x_0^2 & x_1^4x_0 & x_0^3 & x_1x_0^2 & x_1^2x_0 & x_1^3 & x_1^4x_0^4 & x_0 & x_1 & x_1^2x_0^4 & x_1^3x_0^3 & x_1^4x_0^2 & x_0^4 & x_1x_0^3 & x_1^2x_0^2 & x_1^3x_0 & x_1^4 & x_0^2 & x_1x_0 & x_1^2 & x_1^3x_0^4 & x_1^4x_0^3 \\
2\alpha & 0 & x_1^2 & x_1^4 & x_1 & x_1^3 & x_1^2x_0 & x_1^4x_0 & x_1x_0 & x_1^3x_0 & x_0 & x_1^4x_0^2 & x_1x_0^2 & x_1^3x_0^2 & x_0^2 & x_1^2x_0^2 & x_1x_0^3 & x_1^3x_0^3 & x_0^3 & x_1^2x_0^3 & x_1^4x_0^3 & x_1^3x_0^4 & x_0^4 & x_1^2x_0^4 & x_1^4x_0^4 & x_1x_0^4 \\
2 \alpha+1 & 0 & x_1^2x_0 & x_1^4x_0^2 & x_1x_0^3 & x_1^3x_0^4 & x_1^3x_0 & x_0^2 & x_1^2x_0^3 & x_1^4x_0^4 & x_1 & x_1x_0^2 & x_1^3x_0^3 & x_0^4 & x_1^2 & x_1^4x_0 & x_1^4x_0^3 & x_1x_0^4 & x_1^3 & x_0 & x_1^2x_0^2 & x_1^2x_0^4 & x_1^4 & x_1x_0 & x_1^3x_0^2 & x_0^3 \\
2 \alpha+2 & 0 & x_1^2x_0^2 & x_1^4x_0^4 & x_1x_0 & x_1^3x_0^3 & x_1^4x_0 & x_1x_0^3 & x_1^3 & x_0^2 & x_1^2x_0^4 & x_1^3x_0^2 & x_0^4 & x_1^2x_0 & x_1^4x_0^3 & x_1 & x_1^2x_0^3 & x_1^4 & x_1x_0^2 & x_1^3x_0^4 & x_0 & x_1x_0^4 & x_1^3x_0 & x_0^3 & x_1^2 & x_1^4x_0^2 \\
2\alpha+3 & 0 & x_1^2x_0^3 & x_1^4x_0 & x_1x_0^4 & x_1^3x_0^2 & x_0 & x_1^2x_0^4 & x_1^4x_0^2 & x_1 & x_1^3x_0^3 & x_0^2 & x_1^2 & x_1^4x_0^3 & x_1x_0 & x_1^3x_0^4 & x_0^3 & x_1^2x_0 & x_1^4x_0^4 & x_1x_0^2 & x_1^3 & x_0^4 & x_1^2x_0^2 & x_1^4 & x_1x_0^3 & x_1^3x_0 \\
2\alpha+4 & 0 & x_1^2x_0^4 & x_1^4x_0^3 & x_1x_0^2 & x_1^3x_0 & x_1x_0 & x_1^3 & x_0^4 & x_1^2x_0^3 & x_1^4x_0^2 & x_1^2x_0^2 & x_1^4x_0 & x_1 & x_1^3x_0^4 & x_0^3 & x_1^3x_0^3 & x_0^2 & x_1^2x_0 & x_1^4 & x_1x_0^4 & x_1^4x_0^4 & x_1x_0^3 & x_1^3x_0^2 & x_0 & x_1^2 \\
3\alpha & 0 & x_1^3 & x_1 & x_1^4 & x_1^2 & x_1^3x_0^4 & x_1x_0^4 & x_1^4x_0^4 & x_1^2x_0^4 & x_0^4 & x_1x_0^3 & x_1^4x_0^3 & x_1^2x_0^3 & x_0^3 & x_1^3x_0^3 & x_1^4x_0^2 & x_1^2x_0^2 & x_0^2 & x_1^3x_0^2 & x_1x_0^2 & x_1^2x_0 & x_0 & x_1^3x_0 & x_1x_0 & x_1^4x_0 \\
3\alpha+1 & 0 & x_1^3x_0 & x_1x_0^2 & x_1^4x_0^3 & x_1^2x_0^4 & x_1^4x_0^4 & x_1^2 & x_0 & x_1^3x_0^2 & x_1x_0^3 & x_1^3x_0^3 & x_1x_0^4 & x_1^4 & x_1^2x_0 & x_0^2 & x_1^2x_0^2 & x_0^3 & x_1^3x_0^4 & x_1 & x_1^4x_0 & x_1x_0 & x_1^4x_0^2 & x_1^2x_0^3 & x_0^4 & x_1^3 \\
3\alpha+2 & 0 & x_1^3x_0^2 & x_1x_0^4 & x_1^4x_0 & x_1^2x_0^3 & x_0^4 & x_1^3x_0 & x_1x_0^3 & x_1^4 & x_1^2x_0^2 & x_0^3 & x_1^3 & x_1x_0^2 & x_1^4x_0^4 & x_1^2x_0 & x_0^2 & x_1^3x_0^4 & x_1x_0 & x_1^4x_0^3 & x_1^2 & x_0 & x_1^3x_0^3 & x_1 & x_1^4x_0^2 & x_1^2x_0^4 \\
3\alpha+3 & 0 & x_1^3x_0^3 & x_1x_0 & x_1^4x_0^4 & x_1^2x_0^2 & x_1x_0^4 & x_1^4x_0^2 & x_1^2 & x_0^3 & x_1^3x_0 & x_1^2x_0^3 & x_0 & x_1^3x_0^4 & x_1x_0^2 & x_1^4 & x_1^3x_0^2 & x_1 & x_1^4x_0^3 & x_1^2x_0 & x_0^4 & x_1^4x_0 & x_1^2x_0^4 & x_0^2 & x_1^3 & x_1x_0^3 \\
3\alpha+4 & 0 & x_1^3x_0^4 & x_1x_0^3 & x_1^4x_0^2 & x_1^2x_0 & x_1^2x_0^4 & x_0^3 & x_1^3x_0^2 & x_1x_0 & x_1^4 & x_1^4x_0^3 & x_1^2x_0^2 & x_0 & x_1^3 & x_1x_0^4 & x_1x_0^2 & x_1^4x_0 & x_1^2 & x_0^4 & x_1^3x_0^3 & x_1^3x_0 & x_1 & x_1^4x_0^4 & x_1^2x_0^3 & x_0^2 \\
4\alpha & 0 & x_1^4 & x_1^3 & x_1^2 & x_1 & x_1^4x_0^2 & x_1^3x_0^2 & x_1^2x_0^2 & x_1x_0^2 & x_0^2 & x_1^3x_0^4 & x_1^2x_0^4 & x_1x_0^4 & x_0^4 & x_1^4x_0^4 & x_1^2x_0 & x_1x_0 & x_0 & x_1^4x_0 & x_1^3x_0 & x_1x_0^3 & x_0^3 & x_1^4x_0^3 & x_1^3x_0^3 & x_1^2x_0^3 \\
4\alpha+1 & 0 & x_1^4x_0 & x_1^3x_0^2 & x_1^2x_0^3 & x_1x_0^4 & x_0^2 & x_1^4x_0^3 & x_1^3x_0^4 & x_1^2 & x_1x_0 & x_0^4 & x_1^4 & x_1^3x_0 & x_1^2x_0^2 & x_1x_0^3 & x_0 & x_1^4x_0^2 & x_1^3x_0^3 & x_1^2x_0^4 & x_1 & x_0^3 & x_1^4x_0^4 & x_1^3 & x_1^2x_0 & x_1x_0^2 \\
4\alpha+2 & 0 & x_1^4x_0^2 & x_1^3x_0^4 & x_1^2x_0 & x_1x_0^3 & x_1x_0^2 & x_0^4 & x_1^4x_0 & x_1^3x_0^3 & x_1^2 & x_1^2x_0^4 & x_1x_0 & x_0^3 & x_1^4 & x_1^3x_0^2 & x_1^3x_0 & x_1^2x_0^3 & x_1 & x_0^2 & x_1^4x_0^4 & x_1^4x_0^3 & x_1^3 & x_1^2x_0^2 & x_1x_0^4 & x_0 \\
4\alpha+3 & 0 & x_1^4x_0^3 & x_1^3x_0 & x_1^2x_0^4 & x_1x_0^2 & x_1^2x_0^2 & x_1 & x_0^3 & x_1^4x_0 & x_1^3x_0^4 & x_1^4x_0^4 & x_1^3x_0^2 & x_1^2 & x_1x_0^3 & x_0 & x_1x_0 & x_0^4 & x_1^4x_0^2 & x_1^3 & x_1^2x_0^3 & x_1^3x_0^3 & x_1^2x_0 & x_1x_0^4 & x_0^2 & x_1^4 \\
4\alpha+4 & 0 & x_1^4x_0^4 & x_1^3x_0^3 & x_1^2x_0^2 & x_1x_0 & x_1^3x_0^2 & x_1^2x_0 & x_1 & x_0^4 & x_1^4x_0^3 & x_1x_0^4 & x_0^3 & x_1^4x_0^2 & x_1^3x_0 & x_1^2 & x_1^4x_0 & x_1^3 & x_1^2x_0^4 & x_1x_0^3 & x_0^2 & x_1^2x_0^3 & x_1x_0^2 & x_0 & x_1^4 & x_1^3x_0^4 \\
\end{block}
\end{blockarray}
$
}
\caption{The full polynomial parity-check matrix $H(x_0,x_1)$ obtained from starting with the $[5^2,2,24]_{5^2}$ Reed-Solomon code. This means all $25$ values from $\F_{5^2}$ are used as evaluation points, ordered lexicographically as the rows of the matrix. The polynomials in the code with zero constant term are listed in lexicographic ordering as the columns of the matrix.} \label{fig:gf25-matrix}
\end{sidewaysfigure}

\newpage
\bibliographystyle{ieeetr}
\bibliography{IEEEabrv, bib}

\end{document}